%% file: Stateless_Computation.tex
\documentclass[11pt]{article}
\usepackage[utf8]{inputenc}
\usepackage{fullpage}
\input{header.tex}

\begin{document}

\title{Stateless Computation}
\author{Danny Dolev\\
Hebrew University of Jerusalem\\
\texttt{dolev@cs.huji.ac.il}
\and 
Michael Erdmann\\
Carnegie Mellon University\\
\texttt{me@cs.cmu.edu}
\and 
Neil Lutz\\
Rutgers University\\
\texttt{njlutz@rutgers.edu}
\and 
Michael Schapira\\ 
Hebrew University of Jerusalem\\
\texttt{schapiram@cs.huji.ac.il}
\and 
Adva Zair\\
Hebrew University of Jerusalem\\
\texttt{adva.zair@mail.huji.ac.il}
}
\date{\today}

\maketitle

\thispagestyle{empty}

\setcounter{page}{0}

\begin{abstract}
We present and explore a model of \emph{stateless} and \emph{self-stabilizing} distributed computation, inspired by real-world applications such as routing on today's Internet. Processors in our model do not have an internal state, but rather interact by repeatedly mapping incoming messages (``labels'') to outgoing messages and output values. While seemingly too restrictive to be of interest, stateless computation encompasses both classical game-theoretic notions of strategic interaction and a broad range of practical applications (e.g., Internet protocols, circuits, diffusion of technologies in social networks). We embark on a holistic exploration of stateless computation. We tackle two important questions: (1) Under what conditions is self-stabilization, i.e., guaranteed ``convergence'' to a ``legitimate'' global configuration, achievable for stateless computation? and (2) What is the computational power of stateless computation? Our results for self-stabilization include a general necessary condition for self-stabilization and hardness results for verifying that a stateless protocol is self-stabilizing. Our main results for the power of stateless computation show that labels of logarithmic length in the number of processors yield substantial computational power even on ring topologies. We present a separation between unidirectional and bidirectional rings ($\textup{L/poly}$ vs. $\textup{P/poly}$), reflecting the sequential nature of computation on a unidirectional ring, as opposed to the parallelism afforded by the bidirectional ring. We leave the reader with many exciting directions for future research.
\end{abstract} 

\newpage
\section{Introduction}

We put forth a model of \emph{stateless} distributed computation and initiate the systematic exploration of the power of such computation. We first discuss the motivation for this model and then present an exposition of the model and our results.

\subsection{Motivation}
%\todo{Understand if the general (linear in n time) vertex-coloring is also stateless in nature. Read on TDMA protocol as it is vertex-coloring based.}
One of our main motivating examples is routing with the Border Gateway Protocol (BGP). BGP establishes routes between the independently administered networks that make up the Internet, and can be regarded as the glue that holds the Internet together. A BGP-speaking router continuously (1) receives update messages from each of its neighbors, announcing routes to different destinations (IP address prefixes) on the Internet, (2) selects the best available route to each destination (according to its local preferences), and (3) propagates its newly-selected routes to neighboring BGP routers. An important and extensively-studied desideratum is for this decentralized route-selection process to converge to a stable routing configuration from \emph{any} initial configuration of the routing system~\cite{griffin2002stable}. 

While BGP route selection involves maintaining state, in terms of internal memory, this is for the sole purpose of recording the last route-advertisement (per IP prefix) received from each neighbor. Indeed, BGP's route-selection can be modeled as a function that maps the most recent messages received from neighbors to routing choices (see~\cite{griffin2002stable}). BGP is a prominent example for environments in which all nodes repeatedly ``best respond'' to each other's most recent choices of actions. Other examples of ``best-response dynamics''~\cite{hart2005adaptive} include additional network protocols~\cite{karp2000gpsr}, diffusion of technologies in social networks~\cite{morris2000contagion}, circuits with feedback loops, and more, as discussed in~\cite{JSW11,JLSW15, engelberg2013best}. Such environments, and others, fall within our framework of stateless computation, in which computational nodes are modeled as not having an internal state (i.e., memory), but rather interact by repeatedly mapping incoming messages (``labels'') to outgoing messages and output values. Since such applications are inherently ongoing and susceptible to transient faults, we focus on \emph{self-stabilizing} protocols.

\subsection{Modeling Stateless Computation}

Consider a distributed network of size $n$, in which every node (processor) receives an external input, $x_i$. The processors compute a global function, $f(x_1, \ldots, x_n)$, by repeatedly exchanging messages. We consider computation in which nodes have no internal \emph{state}. Instead, each node is equipped with a reaction function that maps incoming messages from neighbors to (1) outgoing messages to neighbors and (2) an output value, based on the node's input. We  abstract the communication model and refer to a message between two nodes as a \emph{label} on the edge connecting them. Edge labels are assumed to come from a finite set $\Sigma$, the \emph{label space}. A protocol is a specification of the label space and the reaction functions. 

Self-stabilizing distributed systems enjoy an important robustness property: the ability to recover from any transient fault, provided that the processor code and input remain intact. The notion of self-stabilization was introduced in a seminal paper by Dijkstra~\cite{dijkstra1982self}. A system is self-stabilizing if, regardless of its initialization, it is guaranteed to arrive at a global ``legitimate state'' in a finite number of steps. We consider two definitions of a legitimate state. From an algorithmic perspective, a legitimate state is a state in which the output value of every node has converged to $f(x_1, \ldots, x_n)$.  A stronger stabilization condition requires not only a correct convergence of the outputs, but also the convergence of the communication between the nodes, i.e., that at convergence all reaction functions be at a fixed point. We call the first condition \emph{output stabilization} and the latter \emph{label stabilization}. %\footnote{In the literature, algorithms that satisfy this condition are sometimes called \emph{silent self-stabilizing}~\cite{dolev1999memory}.} 
In practice, label stabilization can be translated to a reduction in communication overhead and bandwidth. Hence, such a property is clearly appealing in the context of algorithm design for distributed networks.

\subsection{Our Results}

We embark on a systematic exploration of stateless computation. Two sets of questions arise naturally from a distributed computing viewpoint: (1) Under what conditions is self-stabilization of stateless computation guaranteed? and (2) What is the computational power of our model (and its limitations)?

We model the asynchronous nature of distributed environments as a (possibly adversarial) \emph{schedule} that determines which nodes will be \emph{activated} in every time step. Upon activation, a node updates its outgoing labels according to its current incoming label profile, according to its reaction function. We consider the notion of $r$-fair schedule, which captures that each node must be activated at least once in every $r$ consecutive time steps, and we define protocols that converge for every $r$-fair schedule as $r$-stabilizing. Our main general impossibility result (Section~\ref{sec:asynch}) states that simply the existence two stable labelings implies that the protocol cannot be label $(n-1)$-stabilizing. This result imposes an upper bound for the asynchronous robustness of such protocols. We then show that this bound is tight, when we exhibit a protocol that converges for every $r$-fair schedule, for $r < n-1$. As best-response dynamics (with unique best-responses) is a specific realization of stateless computation, our impossibility result implies new nonconvergence results for the broad range of environments that fall within this category, including BGP routing, congestion control, asynchronous circuits, and diffusion of technologies, discussed and formalized in~\cite{JSW11,JLSW15}.

We present in Section~\ref{ssec:hardness} two complementary hardness results, establishing that determining whether a protocol is $r$-stabilizing is both computationally and communicationally hard. We prove that determining whether a protocol is $r$-stabilizing is PSPACE complete, and requires at least exponential number of bits of communication. As these results are proven with respect to \emph{all} possible values of $r$, our proofs are applicable even when assuming synchronized communication (i.e., for $r=1$).

We next turn our attention (in~\cref{sec:os,sec:ls}) to the question of the computational power of stateless computation. We focus on synchronous computation and consider two complexity measures: the \emph{round complexity}, defined as the maximum number of rounds (time units) it takes the protocol to converge, and the \emph{label complexity}, defined as the length of the labels in binary encoding, $\log(|\Sigma|)$. We provide straightforward general upper bounds on the label complexity and round complexity of any function $f$, showing that a linear number of rounds and linear label length are sufficient to compute any function. We show that there exist hard functions that require labels of linear length, matching the general upper bound. We thus investigate what functions can (and cannot) be computed when the label complexity is restricted. 
 %For function $h:\bN\to\bN$, we introduce the complexity classes $\text{OS}^u_h$, $\text{OS}^b_h$, $\text{LS}^u_h$ and $\text{LS}^b_h$ consisting of the languages for which membership can be decided by stateless protocols in  unidirectional and  bidirectional rings using labels of length $h(n)$. In the cases of $\text{LS}^u_h$ and $\text{LS}^b_h$, these protocols are required to be label-stabilizing.

We first examine output-stabilizing protocols in Section~\ref{sec:os}. Our investigation reveals that even in the seemingly simplest network topologies, such as the unidirectional and bidirectional rings, stateless computation is quite powerful. We show that protocols on the unidirectional ring have the same computational power as branching programs of polynomial size, 
%i.e., %$\text{OS}^u_{\log} = 
$\lpoly$. On the bidirectional ring, we show that protocols with polynomial round complexity essentially have the same computational power as Boolean circuits of polynomial size, i.e., they can decide the languages in $\ppoly$. Our results imply that proving super-logarithmic lower bounds in the output-stabilizing scheme is linked to resolving fundamental open
questions in complexity theory.
%, i.e., $\widetilde{\textup{OS}}^b_{\log} = \ppoly$.
%\adva{TODO: we need to decide what notation to use when we mention the ppoly result}

In Section~\ref{sec:ls}, we examine the computational power of label-stabilizing protocols. We first present a general method for proving lower bounds on the label complexity of label-stabilizing protocols on arbitrary graphs, and we  utilize this method to prove a linear lower bound and a logarithmic lower bound on the label complexity of protocols on ring topologies for specific functions (equality and majority, respectively). %We then show that every regular language can be computed on a ring topology using labels of constant length. 

We conclude and discuss directions for future research in Section~\ref{sec:conc}.

\subsection{Related Work}

Our notion of \emph{output stabilization} is the central objective of research on self-stabilization in distributed computing. The key difference from prior research on this topic is the statelessness restriction. Our notion of \emph{label stabilization}, in contrast, corresponds to the more specific notion of \emph{silent} self-stabilizing algorithms~\cite{dolev1999memory}. There is a large body of literature on the design of silent self-stabilizing algorithms for various kinds of tasks (e.g.,~\cite{huang1992self,kosowski2005self,afek1990memory,cournier2009new}). 
Our hardness results therefore translate to results for the self-stabilization of any silent algorithms and our impossibility result translates to impossibility of self-stabilization of \emph{stateless} silent algorithms.  
A widely studied measure in the silent self-stabilization literature is the memory requirements of the nodes' public registers, known as the \emph{compactness} criterion. This measure is analogous to label complexity, so our lower-bounding method from Section~\ref{sec:ls} can be applied to silent self-stabilizing protocols in stateless settings.

%The ring topology is one of the simplest network topologies, yet even this simple topology exhibits interesting phenomena. For this reason, understanding computation on a ring is a good starting point in the exploration of many fundamental problems in distributed networks. Indeed, computation on a ring has been widely studied in a broad variety of contexts, including asynchronous communication, indistinguishable (anonymous) processors, lack of global knowledge (e.g., regarding the network topology and network size), etc.~\cite{attiya1988computing, abrahamson1986probabilistic, burns1980formal, awerbuch1990trade, goldreich1986effects, pachl1982technique, moran1986gap, mansour1986bit}.

%Mansour and Zaks \cite{mansour1986bit} studied a model of a ring with a leader, where all nodes but the leader run the same algorithm. They showed that the languages that can be computed in \emph{bit complexity} of $O(n)$ are precisely the regular languages. Bit complexity measures the (worst case) total number of bits sent by the algorithm. We use the algorithm given in that work for computing regular languages to show that regular languages can be decided by label-stabilizing protocol using a constant number of bits.

Systems in which strategic nodes repeatedly \emph{best respond} to each others' actions are often stateless. Such systems, which include interdomain routing on the Internet, asynchronous circuits, and congestion control, were studied by Jaggard et al.~\cite{JSW11,JLSW15}. The immediate precursor of our work is~\cite{JSW11}, which analyzes convergence of best-response dynamics in asynchronous environments. Our results for non-self-stabilization strengthen and generalize the main results in~\cite{engelberg2013best,JSW11} by extending them to stateless computation in general and to all $r$-fair schedules, closing some of the open questions posed in~\cite{JSW11}. Also, while~\cite{JSW11,JLSW15} focus on asynchronous environments, our investigation also encompasses the computational power of stateless computation.
%Our focus, however, greatly differs from that of~\cite{JLSW15}. While the focus of~\cite{JLSW15} is on studying the behavior of best-response dynamics under asynchrony, our model allows nodes' signals to each other (i.e., messages) to be distinct from their outputs, thus capturing a much richer form of expressiveness, and we analyze the computational power of this framework. Also, whereas~\cite{JLSW15} analyzes best-response dynamics in a clique, we study the computational limitations introduced by a poorly connected topology.

\section{Model and Observations}\label{sec:model}

We consider a model of distributed computation on a  strongly connected directed graph $G=([n],E)$, where each node represents a single processor. Each node $i\in[n]$ has a private input $x_i$ from an input space $X$. Informally, the nodes (processors) are \emph{stateless}, in the sense that a processor cannot store information. A node $i$ can communicate with another node $j$ if and only if the edge $(i,j)$ is in $E$, and this communication is captured by $i$ assigning a \emph{label} to that edge. We formalize the interaction between nodes below.

\subsection{Schedules and Stateless Protocols} 

Let $\Sigma$ be a nonempty finite set that shall henceforth be called the \emph{label space}. Each node $i$ is equipped with a \emph{reaction function} mapping the labels of $i$'s incoming edges and $i$'s input value $x_i$ to an assignment of labels to $i$'s outgoing edges and an output value in $\{0,1\}$. Formally, the reaction function of node $i$ is a deterministic mapping
\[\delta_i: \Sigma^{-i} \times \{0,1\} \to\Sigma^{+i} \times \{0,1\}\,,\]
where $-i$ and $+i$ are i’s sets of incoming and outgoing edges, respectively. %In aggregate, the functions $\delta_i$ define a global transition function
%\[\delta:\Sigma^E\times \{0,1\}^n\to\Sigma^E\times \{0,1\}^n\,,\]
%where $\Sigma^E$ is the set of all \emph{labelings}. %If the processors are consider to be \emph{anonymous}, then all nodes have the same reaction function.

%\subsection{The Asynchronous Model}
%\subsubsection{Schedules}
A schedule is a function, $\sigma:\mathbb{N}^+ \to 2^{[n]}$ that maps time units, $t = 1,2,\ldots$ to a nonempty subset of nodes, $\sigma(t)\subseteq [n]$. $\sigma(t)$ specifies, for each time step $t$, which nodes are \emph{activated} at every time step. Upon activation, a node applies its reaction function to its incoming labels and input and updates its outgoing labels and output accordingly. 
%Formally, let $\sigma(t)$ be the the subset of vertices at time $t$, then for every $j\in \sigma(t)$:
%\[
% \delta_j(\ell_{-i}^{t}, x_j) = (\ell_{+i}^{t+1}, y^{t+1})
%\] 
A schedule is \emph{fair} if for every $j\in [n]$ there exists an infinite sequence of time steps in which $j\in \sigma(t)$. For any positive integer $r$, we say that a schedule is $r$-fair if every node is activated at least once in every sequence of $r$ consecutive time steps.

A \emph{stateless protocol} $A=(\Sigma,\delta)$ specifies the label space and the vector $\delta=(\delta_1,\ldots,\delta_n)$ of individual reaction functions.  Because the reaction functions are deterministic, a global input $(x_1,\ldots,x_n)\in \{0, 1\}^n$, an \emph{initial labeling} $\ell^0$ and a schedule $\sigma$ completely determine all subsequent labelings and outputs. Formally, at every time step $t\in\bN$,  and every $i\in \sigma(t)$, $i$'s outgoing labeling and output are given by:
\[
 (\ell_{+i}^{t}, y_i^{t}) = \delta_i(\ell_{-i}^{t-1}, x_i)
\] 
In aggregate, the functions $\delta_i$ together with the schedule $\sigma$ define a global transition function:
\[
\delta: \Sigma^E\times \{0, 1\}^n\times 2^{[n]}\to \Sigma^E\times \{0, 1\}^n
\]
Satisfy,
\[
(\ell^{t}, y^{t}) = \delta(\ell^{t-1}, x, \sigma(t))
\]

%Formally, at every time step $t\in\bN$, the labeling $\ell^t$ and each node $i$'s output $y_i^t$ are given by
%\[(\ell^t,(y_1^t,\ldots,y_n^t))=\delta(\ell^{t-1},(x_1,\ldots,x_n))\,.\]

\subsection{Self-Stabilization and Computation}

We seek \emph{self-stablizing} protocols, which always converge to a global ``legitimate state'' regardless of initial conditions. We consider two possible formalizations of this concept: \emph{output} $r$-\emph{stabilization} and \emph{label} $r$-\emph{stabilization}. A stateless protocol is output $r$-stabilizing if, for every possible input $(x_1,\ldots,x_n)$ , every initial labeling $\ell^0$, and every $r$-fair scheduling $\sigma$, the output sequence of every node $i$'s, $y_i^1,y_i^2,\ldots$ converges. For a protocol to be label $r$-stabilizing, we additionally require that the sequence of labelings $\ell^0,\ell^1,\ell^2,\ldots$ always converges, i.e., that all reaction functions $\delta_i$ reach a fixed point. We assume that the reaction functions and inputs remain intact throughout the computation and are not subjected to transient faults.

Consider the objective of computing some function $f:X^n\to Y$ on the global input $(x_1,\ldots, x_{n})$ in a stateless manner. We say that an output $r$-stabilizing protocol \emph{computes} a function $f:X^n\to Y$ if the output sequence $y_i^1,y_i^2,\ldots$ of each node $i\in[n]$ converges to $f(x_1,\ldots,x_n)$. We say that a family of protocols $A_1,A_2,\ldots$ $r$-\emph{decides} a language $\mathcal{L}\subseteq \{0,1\}^*$ if for every $n\in\bN$, the protocol $A_{n}$ computes the \emph{indicator function} for $\mathcal{L}$ on $\{0,1\}^n$, i.e., the function $f_n:\{0,1\}^n\rightarrow \{0, 1\}$ defined by $f_n(x) = 1$ if and only if $x\in \mathcal{L}$. 

\subsection{Round Complexity and Label Complexity}\label{sec:cm}

Let $G=([n],E)$ be a directed graph and $A_n$ a self-stabilizing protocol on $G$. We define the \emph{round complexity}, $R_n$, of $A_n$ as the maximum, over all inputs $(x_1, ... , x_n) \in \{0,1\}^n$ and all initial edge labelings $\ell\in\Sigma^E$, of the number of time steps it takes for the protocol to converge. We define the \emph{label complexity}, $L_n$, as $\log(|\Sigma|)$, i.e., the length of a label in binary encoding. Our results regarding the computational power of stateless computation (as opposed to self-stabilization) focus on the scenario of synchronous interaction between nodes, i.e., that the schedule is $1$-fair and so, in every time step, all nodes are activated. %In Sections \cref{sec:cm,sec:os,sec:ls} we discuss synchronous protocols. For ease of presentation we omit the ``1-'' notation, and refer to the two formulations of self-stabilizing protocols as \emph{output-stabilizing} and \emph{label-stabilizing}. 

%\subsection{Bounds on Complexity}\label{sec:bounds} 

We proceed by presenting some observations about the relationship between round complexity and label complexity and a general upper bound on the label complexity of Boolean functions. The omitted proofs appear in Appendix \ref{apdx:bounds}.

\begin{proposition}
\label{prop:radius}
Let $f:\{0, 1\}^n\rightarrow \{0, 1\}$ be a non constant Boolean function on a graph $G=([n],E)$, and let $r$ be the graph's radius. Then, for every output-stabilizing protocol, $r \leq R_n$. 
\end{proposition}

\begin{proposition}\label{prop:bi}
Let $G=([n], E)$ be a directed graph and $A_n=(\Sigma, \delta)$ a stateless protocol. Then, 
\[
R_n \leq |\Sigma|^{|E|} = (2^{L_n})^{|E|}
\]  
\end{proposition}

%We present straightforward 

\begin{proposition}
\label{prop:genupperbd}
Let $G=([n], E)$ be a strongly connected directed graph and $f:\{0, 1\}^n\rightarrow \{0, 1\}$ a Boolean function. There exists a label-stabilizing protocol, $A_n$, that computes $f$, with $L_n=n+1$ and $R_n=2n$.  
\end{proposition}

%One of our main objectives is to understand whether and when these upper bounds are tight. We focus on the label complexity and tackle the question of what can be computed with label complexity that is strictly less than $O(n)$ bits. 

\section*{Part I: Self-Stabilization of Stateless Protocols}

\section{Impossibility Result for Self-Stabilization}\label{sec:asynch}

We explore under what conditions a stateless protocol can be self-stabilizing. We exhibit a necessary condition for self-stabilization. Before presenting this condition, we present required terminology. A \emph{stable labeling} for a protocol $(\Sigma, \delta)$ on a graph $G$ is a labeling $\ell \in\Sigma^{E}$ that is a fixed point of every reaction function $\delta_i$, meaning that $\delta_i(\ell_{-i}, x_i) = (\ell_{+i}, y_i)$ for every node $i$. 

\begin{theorem}\label{thm_impossible}
No stateless protocol with at least two distinct stable labelings is label $(n-1)$-stabilizing.
\end{theorem}

We show that this result is tight in the sense that a system can be label $(n-2)$-stabilizing even if multiple stable labelings exist. Our proof of Theorem~\ref{thm_impossible} utilizes the classical notion of valency argument from distributed computing theory~\cite{fischer1985impossibility}. Applying a valency argument to our context, however, involves new challenges and the introduction of new ideas. Specifically, the proof involves defining a global state-transition space in a manner that captures all system dynamics under label $(n-1)$-fair schedules. A delicate valency argument is then applied to a carefully chosen \emph{subset} of this space so as to obtain the impossibility result. Our result shows that not only do multiple ``stable states'' induce instability, but this is true even under reasonably fair schedules (linear fairness). This both generalizes and strengthens the nonconvergence result in~\cite{JSW11} and, consequently, has immediate implications for game theory, routing and congestion control on the Internet, diffusion of technologies in social networks, and asynchronous circuits.

%\subsection{Non-Convergence Result}\label{ssec:impossib}

\begin{proof}
Let $G=([n], E)$ be a directed graph and $A_n=(\Sigma, \delta)$ a stateless protocol. Suppose $\hat{\ell_1}, \hat{\ell_2}\in \Sigma^{E}$ are two distinct stable labelings, and assume for the sake of contradiction that $A_n$ is label $(n-1)$-stabilizing. For simplicity we ignore the private input and output bits of the nodes, as our only concern is whether or not the labeling sequence stabilizes.

We build a directed \emph{states-graph} $G'=(V', E')$ as follows. Let $m=|E|$ and $r=n-1$. The vertex set is $V' = \Sigma^{m}\times [r]^{n}$. Each node $(\ell, x)\in V'$ consists of a labeling component, $\ell\in \Sigma^{E}$, and a countdown component $x\in [r]^{n}$ that counts for every node in $G$ the number of time steps it is allowed to be inactive. Denote by $V_0' = \{ (\ell, r^{n}) : \ell\in \Sigma^{m}  \}$ the set of initialization vertices. To define the edges set we utilize the following countdown mapping, $c: [r]^n \times 2^{[n]}\rightarrow [r]^n$, that satisfy for every $i\in [n]$:
      \[
      c(x, T)_i = 
      \begin{cases}
        r & \text{If } i \in T\\
        x_i - 1              & \text{Otherwise}
      \end{cases}
      \]
For every node $(\ell, x)$, and every $T \in \{ S :  \{i: x_i=1\} \subseteq S \}$, there is a directed edge from $(\ell, x)$ to $(\delta(\ell, T), c(x, T))$. Observe that every run of the protocol with a $r$-fair schedule is represented in $G'$ as a path from an initial vertex to some subset:
\[ 
T\subseteq \{ (\ell, x): \ell\in \{\hat{\ell}_1, \hat{\ell}_2\}  \text{ , } x\in [r]^n\}
\]. 

We restrict our attention to the subgraph $H = G'(U)$, where:
 \[
 U = \{ (\ell, x): \ell\in\Sigma^{m} \text{ , } s(x) \geq (1, 2, \ldots, r-1, r, r) \},
 \]
and $s(x)$ is the increasingly sorted vector of $x$.

We define the \emph{attractor region} of the stable labeling $\hat{\ell}_1$ as the set of all vertices from which every path reaches a vertex in $\{(\hat{\ell}_1, x ): x\in [r]^n\}$. Namely, from any vertex that is in the attractor region of $\hat{\ell}_1$ we are guaranteed to converge to $\hat{\ell}_1$. We define the attractor region of $\hat{\ell}_2$ as the set of all vertices from which every path reaches a vertex in $\{(\hat{\ell}_i, x ): x\in [r]^n\}$, where $\hat{\ell}_i \neq \hat{\ell}_1$ is a stable state.

To prove that the protocol is not label $r$-stabilizing, we show that existence of a cycle from $i\in V_0'$ (note that $V_0'\subseteq U$) in which every node is not in any attractor region. Assume that $H$ doesn't contain such cycle.
  
\begin{lemma}
\label{lem:VAv0exists}
There exists a node $i\in V_0'$ that does not belong to any attractor regions.
\end{lemma}
\begin{proof}
Suppose every node $i\in V_0'$ is in some attractor region. As there are two attractor regions, at least one node is in the attractor region of $\hat{\ell}_1$ and at least one node is in the attractor region of $\hat{\ell}_2$. Therefore, there also exist two nodes $(\ell, r^n),(\ell',r^n)\in V_0'$ that are in different attractor regions and their labeling components differ in one coordinate, $(i, j)$.
Observe that $(\delta(\ell, \{i\}), c(r^n, \{i\}))=(\delta(\ell', \{i\}) , c(r^n, \{i\}))$. Thus, they both reach the same node, a contradiction.
\end{proof}

\begin{lemma}
\label{lem:VAseq}
Let $(\ell,x), (\ell',y)\in H$ be two vertices satisfying the following:
\begin{enumerate}
\item $\{e: \ell_e\neq \ell'_e\}\subseteq +j$ for some $j\in [n]$.
\item $x \leq y$.
\end{enumerate}
Then there exists an attractor region that is reachable from both.
\end{lemma}
\begin{proof}
Let $(\ell,x)$ and $(\ell',y)$ be such vertices. If $\ell=\ell'$ then observe that $(\delta(\ell,[n]) , c(x, [n]))=(\delta(\ell', [n]) , c(y, [n]))$. Otherwise, let $j$ be the node that some of its outgoing labels differ in $\ell, \ell'$. Let $j'\in \text{argmin}_{\bar{j}\neq j} x_{\bar{j}}$. Consider the vertices $i = (\delta(\ell, \{j,j'\}), c(x, \{j, j'\}))$ and  $i' = (\delta(\ell', \{j,j'\}), c(y, \{j, j'\}))$. Notice that,

\begin{enumerate}
\item For any subset $T$, $c(x, T)\leq c(y, T)$.
\item $(1, 2, \ldots, r-1, r, r) \leq c(x, \{j, j'\})$, thus both vertices are in $H$.
\item Observe that $\ell_{-j} = \ell'_{-j}$ therefore, $\delta(\ell, \{j,j'\})_{+j} = \delta(\ell', \{j,j'\})_{+j}$. 

\item Only if $(j,j')\in E$ the resulted labels may be different, and in this case,
$\{e: \delta(\ell, \{j,j'\})_e\neq \delta(\ell', \{j,j'\})_e\}\subseteq +j'$.
\end{enumerate}
Hence, we can apply the same activation for $i$ and $i'$ and repeat this infinitely many times to create two infinite sequences in $H$ such that the two conditions hold. From our assumption, both sequences eventually reach a stable labeling. As both labeling differ only at the outgoing labels of some node and since activating this node lead to the same labeling, it must be that they converge to the same labeling. 

\end{proof}

\begin{lemma}
\label{lem:VAimp3}
There exists a node $i\in U$, such that $i$ is not in any attractor region, and for every $(i,i')\in E(H)$, $i'$ is in some attractor region.
\end{lemma}
\begin{proof}
Assume that there is no such node. Then any node that is not in any attractor region has a neighbor that is also not in any attractor region. Thus, starting with arbitrary such node (which we know exists, from Lemma~\ref{lem:VAv0exists}) we can create a path that contains only such nodes. But since our graph is finite this path must be a cycle, a contradiction to our assumption.  
\end{proof}

Let $(\ell,x)\in V(H)$ as in Lemma~\ref{lem:VAimp3}. For every $T$ that correspond to an outgoing edge from $(\ell,x)$ the following conditions follow directly from the definition of $H$: (1) $|T|\geq 2$, and (2) If there exists $j$ s.t. $x_j=1$ (there is at most one such vertex) then $j\in T$.

Thus, its edges set is either $2^{[n]}- \bigcup_i\{i\}-\emptyset$ or $\{T\cup\{j\}: T\in 2^{[n] -\{j\}}-\emptyset\}$. Since $(\ell,x)$ is not in any attractor region but its neighbors are either in the attractor region of $\hat{\ell}_1$ or $\hat{\ell}_2$, and combining with our observation that its neighbor set is either the power set of $[n]$ or $[n]\setminus \{j\}$, there must exist $i\in [n]$ and $S\subseteq [n]-\{i\}$ such that $(\delta(\ell, S), c(x, S))$ and $( \delta(\ell, S\cup\{i\}), c(x, S\cup\{i\}))$ belong to two different attractor regions. 

Observe that  $\delta(\ell, S)$ and $ \delta(\ell, S\cup\{i\})$ might differ only in $+i$, and that $c(x, S)\leq c(x, S\cup\{i\})$ so the conditions of Lemma~\ref{lem:VAseq} hold.
Applying the lemma, we get that there is an attractor region that is reachable from both, a contradiction.

\end{proof}

\noindent{\bf Tightness of Theorem~\ref{thm_impossible}.} We next show, via an example, that the impossiblity result of this section is tight.

\begin{example}
We show a protocol over the clique $K_n$ with label space $\Sigma = \{0, 1\}$. For each node $i\in[n]$, the reaction function is
\[
\delta_i(\ell) = 
\begin{cases}
0^{n-1} & \text{if every incoming edge is labeled }0  \\
1^{n-1}              & \text{otherwise}\,.
\end{cases}
\]
Observe that $0^{n(n-1)}$ and $1^{n(n-1)}$ are both stable labelings. Also observe that if at some time step there is more than one node $i$ whose outgoing edges are all labeled $1$, then the labeling sequence is guaranteed to converge to $1^{n(n-1)}$. Hence, an oscillation occurs only if at every time step exactly one node's outgoing edges are all labeled 1. This implies that for an oscillation to occur, (1) two nodes must be activated at each time step, and (2) if $i$'s outgoing edges area labeled 1 at time $t$, then $i$ must be activated at time $t+1$. No $r$-fair schedule for $r<n-1$ can satisfy these two conditions, so the labeling sequence must converge. 
\end{example}

\noindent{\bf Implications for games, routing, circuits, and more.} We point out that best-response dynamics can be formalized in our model as the scenario that both the output set of each node and the labels of each of its outgoing edges are the same set and represent that node's (player's) possible strategies. Thus, a corollary of Theorem~\ref{thm_impossible} is a generalization of the main result in~\cite{JSW11} (Theorem 4.1) for this setting, showing that multiple equilibria imply instability even under linearly-fair schedules. Consequently, Theorem~\ref{thm_impossible} immediately implies strong nonconvergence results for the spectrum of environments formalized in~\cite{JSW11}  (Section 5): routing and congestion control on the Internet, asynchronous circuits, and diffusion of technologies in social networks. We refer the reader to~\cite{JSW11,JLSW15} for more details.

\section{Complexity of Verifying Self-Stabilization} \label{ssec:hardness}

%Throughout this section, we define protocols in which:  (1) The graph family considered is $G=K_n$, (2) There is no function to compute, and (3) All reaction functions output the same outgoing label for all $n-1$ neighbors. Thus we refer to node $i$'s outgoing labeling as single label, $l_i$, and define the reaction function from $n-1$ labels to one label: 
%\[
%l_i^{t+1} = \delta_i(l_1, \ldots, l_{i-1}, l_{i+1}, \ldots, l_n)
%\]

	%\subsection{Computational Complexity of Checking Label Stabilization}
    
We now turn our attention to the complexity of deciding whether a protocol is $r$-stabilizing. We present two complementary results for the communication complexity and the computational complexity models. Our first result establishes that the communication complexity of verifying $r$-stabilization is exponential in the number of nodes $n$.

\begin{theorem}\label{thm:cc}
Let $A_n = (\Sigma,(\delta_1, \ldots, \delta_n))$ be a stateless protocol. Consider the following 2-party communication problem. Alice receives as input a description of $\delta_1$ and Bob receives $\delta_2$. They both have access to $\delta_3, \ldots, \delta_n $, and they need to decide whether $A_n$ is label $r$-stabilizing. For eny value of $r$, deciding whether $A_n$ is label $r$-stabilizing requires $2^{\Omega(n)}$ bits of communication.    
\end{theorem}

Our proof utilizes combinatorial ``snake-in-the-box" constructions, as in~\cite{JSW11}. To prove the theorem for every possible value of $r$ two separate reductions, corresponding to two regimes of $r$ (``higher'' and ``lower'' values) are presented. The proof appears in Appendix~\ref{apdx:hardness}.

The above communication complexity hardness result requires the representation of the reaction
functions to (potentially) be exponentially long. What if the reaction functions can be succinctly
described? We complement the above result present a strong computational complexity hardness result for the case that each reaction function is given explicitly in the form of a Boolean circuit.

\begin{theorem}\label{thm:pspace}
For every $r$, deciding whether a stateless protocol is label $r$-stabilizing is PSPACE-complete.
\end{theorem}

Our proof of Theorem~\ref{thm:pspace} relies on a ``self-stabilization-preserving reduction''~\cite{engelberg2013best}. We first show that the \textsc{String-oscillation} problem can be reduced to the problem of deciding whether a \emph{stateful protocol}, in which reaction functions may depend on both their incoming and their \emph{outgoing} labels, is label $r$-stabilizing. Next, we show how to construct a stateless protocol from a stateful protocol without altering its stabilization properties.  The proof appears in Appendix~\ref{apdx:hardness}.
		
	%\subsection{Communication Complexity of Checking Label Stabilization}

\section*{Part II: Computational Power of Stateless Protocols}

\section{Output Stabilization, Branching Programs, and Circuits}\label{sec:os}

We exhibit output-stabilizing protocols and inspect their computational power in comparison to label-stabilizing algorithms. Consider the clique topology, $K_n$. Note that every Boolean function can be computed using a 1-bit label and withing one round. A similar argument can be made for the \emph{star} topology. We therefore examine poorly connected topologies to gain a better understanding on the computational power and limitations of our model. As a first step in this direction, we consider the \emph{unidirectional} and the \emph{bidirectional} rings.

The ring topology is one of the simplest network topologies, yet even this simple topology exhibits interesting and highly nontrivial phenomena. For this reason,
%For this reason, understanding computation on a ring is a good starting point in the exploration of many fundamental problems in distributed networks. Indeed,
computation on a ring has been widely studied in a broad variety of contexts, including asynchronous communication, indistinguishable processors, lack of global knowledge, etc.~\cite{attiya1988computing, abrahamson1986probabilistic, burns1980formal, awerbuch1990trade, goldreich1986effects, pachl1982technique, moran1986gap, mansour1986bit}.

Our main positive results establish that even when considering a system with limited expressive power, in terms of relatively succinct label size (logarithmic in $n$), all problems that can be solved efficiently via centralized computation, can also be solved in our restricted model.  

We obtain a separation between protocols on the unidirectional and bidirectional rings through an exact characterization of the computational power of each of these environments. 
We show that protocols on the unidirectional ring have the same computational power as branching programs of polynomial size, i.e., they can compute the languages in $\lpoly$. This is also the class of all languages that can be decided by a logspace Turing machines that receive, along with each input $\{0,1\}^n$, an auxiliary \emph{advice} string that depends only on $n$ and is of length polynomial in $n$. This advice is given on a secondary read-only tape, separate from the input tape and work tape.

Our second theorem shows that protocols on the bidirectional ring are stronger, and have the same computational power as polynomial-sized Boolean circuits, the class $\ppoly$. This is also the class of all languages that can be decided by polynomial time Turing machines with a polynomial sized advice.

Our characterization results reveal that the expressiveness for bidirectional rings is dramatically richer. This, in some sense, reflects the sequential nature of computation on a unidirectional ring, as opposed to the parallelism afforded by the bidirectional ring. We explain below that these results imply that proving super-logarithmic lower bounds for the unidirectional or bidirectional rings would imply major breakthroughs in complexity theory.

\begin{definition}
	Given any function $h:\bN\to\bN$, a language $\mathcal{L}\subseteq\{0,1\}^*$ belongs to the class $\textup{OS}^u_{h}$ if $\mathcal{L}$ is decided by some family $A_1,A_2\ldots$ of stateless protocols such that each $A_n$ has label complexity $O(h(n))$ and runs on the unidirectional $n$-ring. %If each $A_n$ is also label-stabilizing, then $\mathcal{L}$ belongs to the class $\textup{LS}^{u}_{h}$.
\end{definition}
%We make two simple observations regarding these complexity classes. First, for any function $h:\bN\to\bN$, the inclusion $\text{OS}^{u}_{h}\subseteq\text{OS}^{b}_{h}$ follows immediately from our definition. Second, Proposition~\ref{prop:genupperbd} implies that if $h(n)=\Omega(n)$ then every language belongs to $\text{LS}_h^u$. 

\begin{theorem}\label{thm:lpoly}
	$\textup{OS}^u_{\log}\equiv \lpoly$.
\end{theorem}

While we do not have a complete characterization for the class $OS^b_{\log}$, we can characterize a simple variant of this class. If we allow polynomially many ``helper'' nodes whose inputs do not affect the function value, then bidirectional protocols with logarithmic label complexity and polynomial round complexity have the same computational power as circuits of polynomial size. 
 
% OPTION 1:
\begin{definition}
	Given any function $h:\bN\to\bN$, a language $\mathcal{L}\subseteq\{0,1\}^*$ belongs to the class $\widetilde{\textup{OS}}^b_{h}$ if $\mathcal{L}$ is decided by some family $A_1,A_2\ldots$ of stateless protocols such that each $A_n$ has polynomial round complexity, label complexity $O(h(n))$, and runs on the bidirectional $p(n)$-ring, where $p(n)$ is polynomial in $n$.
\end{definition}

\begin{theorem}\label{thm:ppoly}
	$\widetilde{\textup{OS}}^b_{\log} \equiv \ppoly$.
%  A language $\mathcal{L}$ is in $\ppoly$ if and only if there is some polynomial $p:\bN\to\bN$ such that the padded language $\mathcal{L}'=\{x0^{p(|x|)} : x\in\mathcal{L}\}$ can be computed by a family of bidirectional protocols with logarithmic label complexity and polynomial round complexity.
\end{theorem}

The complete proof of Theorem~\ref{thm:ppoly} is given in Appendix~\ref{apdx:os}. The central idea in simulating a Boolean Circuit in the bidirectional ring is to enable the nodes to simultaneously count using a synchronous $D$-counter. The way we think of counting is that after ``sufficiently long'' time has passed, all nodes simultanesouly see the same sequence of incoming labels with the numbers $0, 1, 2, \ldots$, repeated indefinitely. We now show a stateless $D$-counter protocol with a label complexity of $L_n = O(\log(D))$. We point out that this protocol does not compute any function, only reaches this desired global sequence of states. Thus, our reaction functions capture a mapping from incoming labels to outgoing labels only. The label complexity of simulating a global counter depends \emph{only} on the counting value, $D$. Our protocol utilize a $2$-counter protocol as a building block. We first present the $2$-counter protocol followed by a $D$-counter protocol.

%For the other direction, we use two protocols that assist us simulating a Boolean circuit. The first is a counter protocol, that we use to carry out a global clock in a form of a counter that counts up to some value, $D$, and repeats it indefinitely. The second, shows how a communication between one node and its neighbor can serve as a method to retain memory throughout the execution. We point out that this protocol doesn't compute any function, and merely reaches this desired global state, thus our reaction function depict a mapping from incoming labels to outgoing labels. The label complexity of simulating a global counter depends \emph{only} on the counting value, $D$.

%We begin with introducing a counter protocol for odd-sized bidirectional rings, in which all nodes count periodically from $0$ to $D-1$, with a label complexity of $L_n = O(\log(D))$. By counting we mean that at stabilization, all nodes receive from their neighbors the value $0$, then, at the next time step, $1$, and so on until $D-1$, after which they see $0$ again etc. We point out that this protocol doesn't compute any function, and merely reaches this desired global state regardless of the nodes input/output. The label complexity of simulating a global clock depends \emph{only} on the counter value, $D$. %The proof appears in Appendix \ref{apdx:os}.

%\vspace{-1em}
%\begin{proof}%[Proof of Claim \ref{clm:counter}] 
%The implementation of the counter protocol relies on the following lemma.
\begin{claim}
\label{lem:flipping}
For every odd-sized bidirectional $n$-ring there exists a stateless $2$-counter protocol.
%$A_n$ over the label space $\Sigma = \{0, 1\}$ with the following
%property. At convergence, there is (at least) one node that receives
%an infinite sequence of flipping bits, i.e., $(0, 1, 0, 1, ...)$.
\end{claim}
\begin{proof}
We begin with a proof for $n=3$. 
The label space is $\Sigma=\{0, 1\}^2$.
All reaction functions send the same
label in both directions, so we refer to their outgoing labels as two
bits. We use the notation $\ell_{(i, j)}[i]$ to denote the $i^{th}$ bit in label the $\ell_{(i, j)}\in \Sigma$. %We first explain how the reaction functions map the first bit, and then how they map the second bit. 
The reaction function of node $1$, 
\[
\delta_1(\ell_{(2, 1)}, \ell_{(3, 1)})[1] = \text{NOT}(\ell_{(2, 1)}[1])
\]
\[
\delta_1(\ell_{(2, 1)}, \ell_{(3, 1)})[2] = \ell_{(3, 1)}[1]
\]
The reaction function of node $2$, 
\[
\delta_2(\ell_{(1, 2)}, \ell_{(3, 2)})[1] = \ell_{(1, 2)})[1]
\]
\[
\delta_2(\ell_{(1, 2)}, \ell_{(3, 2)})[2] = \text{NOT}(\ell_{(1, 2)}[2])
\]
The reaction function of node $3$ is a XOR over the first bit of its incoming labels:
\[
\delta_3(\ell_{(1, 3)}, \ell_{(2, 3)})[1] = 
\text{XOR}(\ell_{(1, 3)}[1],\ell_{(2, 3)}[1])
\]
\[
\delta_3(\ell_{(1, 3)}, \ell_{(2, 3)})[2] = \ell_{(2, 3)}[2]
\]

%\[
%\delta_1(\ell_{(2, 1)}, \ell_{(3, 1)})[2] = \ell_{(3,1)}[2]
%\]
%The first reaction function, $\delta_1$, sends $1$ iff its incoming labels are either $(0, 1)$ or $(1, 0)$. 
%$\delta_2$, send $1$, iff $\ell_{(3,2)}=1$. 

%$\delta_3$, send $1$ iff $\ell_{(2,3)}=0$. 

Assume that $(\ell_{(2,1)}[1], \ell_{(1,2)}[1]) = (0,0)$. Observe that at the next time step they alter to $(0, 1)$, then $(1, 1)$, $(1, 0)$ and again, $(0,0)$. Note that they agree/disagree on their outgoing labels every other time step,
therefore, node $3$ must send alternating bits, in its first label bit, and they converge after at most two time steps. 

To generalize for any odd $n$, we assign the nodes' reaction functions $\delta'_1,\delta'_2, \ldots , \delta'_n$ as follow. $\delta'_n = \delta_3$,
$\delta'_{1} = \delta_1$. 
For every other node $1<j<n$, if $j$ is even then 
\[
\delta'_j(\ell_{(j-1, j)}, \ell_{(j+1, j)})[1] = \ell_{(j-1, j)})[1]
\]
\[
\delta'_j(\ell_{(j-1, j)}, \ell_{(j+1, j)})[2] = \text{NOT}(\ell_{(j-1, j)}[2])
\]
and if $j$ is odd then,
\[
\delta'_j(\ell_{(j-1, j)}, \ell_{(j+1, j)})[1] = \ell_{(j-1, j)})[1]
\]
\[
\delta'_j(\ell_{(j-1, j)}, \ell_{(j+1, j)})[2] = \ell_{(j-1, j)}[2])
\]
Notice that the outgoing first bit label sequence of $1$ is $(0, 0, 1, 1, 0, 0, \ldots)$. Since node $2$ repeats it, its first label bit sequence will be the same but with a delay of one time step. As the number of nodes is odd, the delay of node $n-1$ must be odd too, thus nodes $1, n-1$ agree/disagree every other time step, and node $n$ outputs alternate bits at the first bit label sequence. Since node $1$ updates its outgoing second bit according to node $n$ first bit, and it will output alternating bits at its second bit sequence. Node $2$ negating the value it sees at the second bit thus, they both observe the same bit at every time step. We complete the proof by induction over $n$, while using the fact $n$ is odd.
\end{proof}

\begin{claim}\label{clm:counter}
For every odd-sized bidirectional $n$-ring there exists a stateless $D$-counter protocol.
\end{claim}
\begin{proof}
Suppose we have a ring of size $n=2$, the label space $\Sigma =
\{0, 1, ... , D-1\}$, and nodes $1$ and $2$ have the same reaction
function that increments the value on their incoming label, modulo $D$.
Let $(\ell_{(1, 2)}, \ell_{(2, 1)}) = (x,y)$ be the initial labeling. 
After one time step, it will turn to $(y+1, x+1)$, after two time steps, $(x+2, y+2)$, and in general the sequence of incoming labels of node $1$ is 
\[(y, x+1, y+2, x+3, ...)\] 
and of node $2$ is 
\[(x, y+1, x+2, y+3, ...)\]
Now, assume
they both knew $x-y$, then by adding it to the $y$ subsequence, they
were able to count together. Clearly, they also need to be able to distinct between the $``x''$ and $``y''$ subsequences. Based on this observation, we will describe the counter algorithm.

\begin{enumerate}
\item The label space is $\Sigma = \{ (b_1, b_2, z, g, c) \} $, where:
\begin{enumerate}
\item $b_1$ and $b_2$ the two bits that implement the $2$-counter protocol from Claim~\ref{lem:flipping}.
\item $z, g, c$ are integers from $[D]$.
\end{enumerate}

\item The reaction function of node $1$ assigns $z=x+1$ for $\ell_{(2, 1)} = (-, -, x, -, -)$. 
\item Every other node $j$, assigns $z=y+1$ for $\ell_{(j-1, j)} = (-, -, y, -,
-)$.
\item Observe that all odd/even nodes have the same outgoing sequence, thus, node $1$ can see at the same time the two subsequences in the $z$ field, and calculate the gap $g=x-y$. Then, it simply propagate it to all nodes in a clockwise direction in the $g$ field. 
\item Finally, the counter value $c$ will be for every odd node as 
\[
c = 
\begin{cases}
z+1+g & \text{If } b_2=0 \\
z+1 & \text{Otherwise}
\end{cases}
\]
and for even nodes as
\[
c = 
\begin{cases}
z+1 & \text{If } b_2=0 \\
z+1+g & \text{Otherwise}
\end{cases}
\]
%defined as $c=z+1+g$ if $b_2$ is $0$ and $c=z+1$, otherwise.
\end{enumerate}
 
%The first bit $b_1$ will be used to implement the protocol from\lemref{lem:flipping}.
%Without loss of generality,
%suppose that after the protocol converged the incoming labels sequence of node $1$  contain flipping bits. All other nodes, will propagate this bit in the clockwise
%direction using $b_2$. As a result, if the incoming label of node $1$ has
%$b_1=\beta$, the incoming label of all odd vertices has $b_2 = \beta$,
%whereas in all even vertices $b_2=1-\beta$. 
The round complexity of the algorithm is $R_n = 4n$ as $n$ time steps are needed for each field $b_1, b_2, z$ and $g$ to stabilize. The label complexity is $L_n = 2+3\log(D)$.
\end{proof}

\begin{corollary}
\label{coro:1}
$P\subseteq \logBOS$ and $L\subseteq \textup{OS}^u_{\log}$
\end{corollary}

\begin{corollary}
\label{coro:2}
If $\mathcal{L} \in P$ but $\mathcal{L} \notin \textup{OS}^u_{\log}$, then $P \neq LSPACE$
\end{corollary}

\begin{corollary}
\label{coro:3}
If $\mathcal{L}\in \text{NP}$ but $\mathcal{L}\notin \logBOS$, then $P\neq NP$
\end{corollary}

%\section{Hard Functions Exist}\label{sec:hardFunc}
It follows from the above corollaries that: (1) obtaining super logarithmic lower bounds on the label complexity both for either unidirectional or bidirectional rings would imply major breakthroughts in complexity theory and is thus challenging, and (2) every efficiently-decidable language can be computed in a decentralized, stateless, manner with low label complexity. 
Nevertheless, a simple counting argument reveals that some functions cannot be computed by protocols with sublinear label complexity. In fact, this is true for any network topology in which the maximum degree is constant, and even if we do not require label stabilization. This implies that there are problems for which we cannot hope to achieve a better complexity than the trivial upper bound of Proposition \ref{prop:genupperbd}. 

\begin{theorem}\label{thm:hardFunc}
Let $\{G_n\}_{n=1}^{\infty}$ be graph family, so that the maximal degree  of $G_n$ is  constant, i.e., there is some $k\in \mathbb{N}$ so that for every $n$, $\Delta(G_n) = k$. Then for every $n>8$, there exists a function $f:\{0,
1\}^n\rightarrow \{0, 1\}$ that cannot be computed by a stateless protocol on $G_n$ with $L_n < n/(4k)$.
\end{theorem}
\begin{proof}
The number of possible protocols over the label space $\Sigma$ is at most $(2|\Sigma|^k)^{2n|\Sigma|^k}$.
Applying the same protocol for two different Boolean functions on an
inconsistent input must result in an incorrect behavior. It follows that the numbers of protocols is at least the number of Boolean functions, $2^{2^n}$:
\[
(2|\Sigma|^k)^{2n|\Sigma|^k} \geq 2^{2^{n}}
\] 
\[
2n|\Sigma|^k \cdot \log(2|\Sigma|^k) \geq 2^n
\]
\[
\log(n) +\log(2|\Sigma|^k)+\log( \log(2|\Sigma|^k) )\geq n
\]
\[
\log(n)+2k\log(|\Sigma|) \geq n
\]
\[
2k\log(|\Sigma|) \geq n-\log(n)
\]
\[
\log(|\Sigma|) \geq \frac{1}{2k} (n-\log(n)) \geq \frac{1}{2k} (n-\frac{n}{2})
\]    
\[
L_n \geq \frac{n}{4k}\;.
\]

\end{proof}

\section{A Method for Label-stabilizing Lower Bounds}\label{sec:ls}

In this section, we present a method for deriving lower bounds on label-stabilizing protocols on an arbitrary topology. We then conclude linear and logarithmic lower bounds for concrete functions (\emph{equality} and \emph{majority}) on the ring topology. 

%Next, we show that regular languages can be computed by label-stabilizing protocols on a ring with constant label size and obtain a separation between anonymous and ID-based systems in label-stabilizing protocols. 

%\subsection{A Method for Label-stabilizing Lower Bounds}\label{sec:lsbound}

We now show that a variation on the \emph{fooling-set method}~\cite{arora2009computational}, can be used to prove lower bounds for label complexity. We present the formal definition followed by our lower bound theorem, which is proven in Appendix~\ref{apdx:lsbound}.

\begin{definition}

A \emph{fooling set} for a function $f:\{0,1\}^{n} \rightarrow \{0,1\}$ is a set $S\subseteq \{0,1\}^{m}\times \{0,1\}^{n-m}$, for some $m\leq n$, such that (1) for every $(x,y)\in S$, $f(x,y) = b$, and (2) for every distinct $(x,y),(x',y')\in S$, either $f(x ,y')\neq b$ or $f(x', y)\neq b$. \end{definition} 

\begin{theorem}
\label{thm:lowerbound}
	Let $m\in\bN$, let $f:\{0, 1\}^{n} \rightarrow \{0, 1\}$ be a function, and let $G=([n],E)$ be a directed graph. Define $C=\{(i,j)\in E:i\leq m<j\}$ and $D=\{(i,j)\in E:j\leq m<i\}$,
	the sets of edges out of and into the node subset $[m]$. Suppose $f$ has a fooling set $S\subseteq\{0,1\}^{m}\times\{0,1\}^{n-m}$ such that, for every $(x,y),(x^\prime,y^\prime)\in S$,
	\begin{itemize}
		\item if $(i,j)\in C$, then $x_{i} = x^\prime_{i}$, and
		\item if $(i,j)\in D$, then $y_i = y^\prime_i$.
	\end{itemize}
	Then every label-stabilizing protocol on $G$ that computes $f$ has label complexity at least $\frac{\log(|S|)}{|C|+|D|}$.
\end{theorem}

We apply this method to give lower bounds on the label complexity of label-stabilizing protocols for two important functions: the \emph{equality} function $\textsc{Eq}_n:\{0, 1\}^{n}\to\{0, 1\}$ defined by $\textsc{Eq}_n(x)=1$ if and only if $n$ is even and $(x_1,\ldots,x_{n/2})=(x_{n/2+1},\ldots,x_{n})$, and the \emph{majority} function $\textsc{Maj}_n:\{0, 1\}^n\to\{0,1\}$ defined by $\textsc{Maj}_{n}(x)=1$ if and only if $\Sigma_{i=1}^{n}x_i \geq n/2$. The proofs appear in Appendix \ref{apdx:lsbound}.

\begin{corollary}\label{cor:eq}
Every label-stabilizing protocol that computes $\textsc{EQ}_n$ on the bidirectional $n$-ring has label complexity at least $\frac{n-2}{8}$.
\end{corollary}

\begin{corollary}\label{cor:maj}
Every label-stabilizing protocol that computes $\textsc{Maj}_{n}$ on the bidirectional $n$-ring has label complexity at least $\log(\lfloor n/2\rfloor)/4$.
\end{corollary}

\section{Summary and Future Work} \label{sec:conc}

We formalized and explored the notion of stateless computation. Our investigation of stateless computation focused on two important aspects: (1) guaranteeing and verifying self-stabilization of stateless computation systems and (2) analyzing the computational power of such systems.

We presented a general necessary condition for self-stabilization, establishing that multiple stable labelings induce potential instability even under relatively fair schedules on node activation. We also investigated the communication and computational complexity of verifying that a system is self-stabilizing and exhibited hardness results for both models of computation. Our results for the computational power of stateless systems show that stateless computation is sufficiently powerful to solve a nontrival range computational problems in a manner that is both robust to transient faults and frugal in terms of message length. We provided exact characterizations of the languages that can be computed by output-stabilizing
protocols with logarithmic label size, on  unidirectional and bidirectional rings.
 
We believe that we have but scratched the surface in our exploration of stateless computation and leave the reader with many open questions: (1) Identifying necessary and sufficient conditions on the reaction functions for self-stabilization. (2) Exploring ``almost stateless'' computation, e.g., computation with aconstant number of internal memory bits. (3) Extending our research to specific network topologies such as the hypercube, torus, trees, etc. (4) Understanding the implications of randomized reaction functions for self-stabilization and computation.

\newpage
\bibliography{bib} 
\bibliographystyle{plain}

\newpage

\appendix

\noindent {\Large \bf Appendix}

\section{Proofs for Section~\ref{sec:model}}\label{apdx:bounds}

\begin{proposition}
%\label{prop:radius}
Let $f:\{0, 1\}^n\rightarrow \{0, 1\}$ be a non constant Boolean function on a graph $G=([n],E)$, and let $r$ be the graph radius. Then for every output-stabilizing protocol, $r \leq R_n$. 
\end{proposition}

\begin{proof}[Proof of Proposition \ref{prop:radius}]
Since $f$ is not constant, there exist two input assignments, $x, x'\in \{0, 1\}^n$, so that:
\begin{enumerate}
\item $x$ and $x'$ differ in exactly one coordinate, denoted by $j$.
\item $f(x)\neq f(x')$.
\end{enumerate}
We assign the graph with $x$ and an arbitrary labeling $\ell^0$, and run $A_n$ $R_n$ time steps (until convergent). We initialize the graph with $x'$ and $\ell^{R_n}$ and run $A_n$ again. It takes at least one time step for node $j$ to update its own output to $f(x')$, at least two time steps for $j$'s neighbors to update their own output, and at least $r$ time steps for output convergence to of all nodes.   
\end{proof}

\begin{proposition}%\label{prop:bi}
Let $G=([n], E)$ be a directed graph and $A_n=(\Sigma, \delta)$, a stateless protocol. Then, 
\[
R_n \leq |\Sigma|^{|E|} = (2^{L_n})^{|E|}
\]  
\end{proposition}
\begin{proof}[Proof of Proposition \ref{prop:bi}]
The maximal time steps until convergence cannot exceed the number of possible system configurations. 
\end{proof}

\begin{proposition}
%\label{prop:genupperbd}
Let $G=([n], E)$ be a strongly connected directed graph and $f:\{0, 1\}^n\rightarrow \{0, 1\}$ a Boolean function. There exists a label-stabilizing protocol, $A_n$, that computes $f$, with $L_n=n+1$ and $R_n=2n$.  
\end{proposition}

\begin{proof}[Proof of Proposition \ref{prop:genupperbd}] 
Since $G$ is strongly-connected there are two spanning trees $T_1, T_2$ rooted in node $1$. In $T_1$ there is a directed path from $1$ to every node $i\neq 1$, and in $T_2$ there is a directed path from every node $i\neq 1$ to $1$. For a node $i\neq 1$ denote by $p_1(i), p_2(i)\in [n]$ its parents nodes in $T_1$ and $T_2$, respectively, and by $c_1(i), c_2(i)\subseteq [n]$, its children sets in $T_1, T_2$. %the parents 
%Thus, each node $i\neq 1$ has an incoming edge $(i, j_-)\in T_1$ and an outgoing edge $(i, j_+)\in T_2$.  
The label space is $\Sigma = \{0,1\}^{n+1}$. A label is composed of two components, $(z, b)$ where $z\in\{0, 1\}^n$ corresponds to the global input vector, and $b\in\{0, 1\}$ corresponds to the global output bit. For every $i\in [n]$, let $w_i\in \{0, 1\}^n$ denote the vector of all zeros except of the $i^{th}$ coordinate, at which it is $x_i$. Let $\text{OR}(z_{c_2(i)})$ be the coordinate-wise OR function of all the $z$ components in the incoming labels from all of $i$'s children in $T_2$.   

The reaction function of node $1$ satisfies $(\ell^{t}_{+1}, y_1^{t}) = \delta_1(\ell_{-1}^{t-1}, x_1)$, where:
\[
y_1^t = f(w_1\lor \text{OR}(z_{c_2(1)}^{t-1}))
\]
and
\[
\ell^t_{(1, j)} = 
\begin{cases}
(0^n, f(w_1\lor \text{OR}(z_{c_2(1)}^{t-1})) ) & \text{If } j\in c_1(1)\\
0^{n+1} & \text{Otherwise},
\end{cases}
\]
The reaction function of every node $i\neq 1$ satisfies $(\ell_{+i}^t, y_i^t) = \delta_i(\ell_{-i}^{t-1}, x_i))$, where:
\[ 
y_i^t = b_{p_1(i)}
\]
\[
\ell^t_{(i, j)} = 
\begin{cases}
(w_i \lor \text{OR}(z_{c_2(i)}^{t-1}), b_{p_1(i)} )& \text{If }j=p_2(i) \text{ and } j \in c_1(i) \\
(0^{n}, b_{p_1(i)} ) & \text{If }j \neq p_2(i) \text{ and } j \in c_1(i) \\ 
( w_i \lor \text{OR}(z_{c_2(i)}^{t-1}), 0) & \text{If }j = p_2(i) \text{ and } j \notin c_1(i) \\ 
0^{n+1} & \text{Otherwise}
\end{cases}
\]

	Informally, the nodes in $T_2$ propagate to node $1$ their own input together with the input of all nodes in their sub-tree in $T_2$. Node $1$ receives this information from its children, and together with its own private input, can calculate $f(x_1,\ldots, x_n)$. Then, node $1$ propagate the resulting bit in $T_1$.
    Note that within at most $n$ time steps node $1$ has a complete knowledge of $(x_1, \ldots, x_n)$, and within at most additional $n$ time steps, $f(x_1, \ldots, x_n)$ is propagated to every node in $G$. 
\end{proof}

\iffalse
\begin{proof}[Proof of Proposition \ref{prop:anupperbd}] 
First observe that a necessary condition for a function $f$ to be computable on an anonymous ring is that $f$ must be invariant to cyclic
shifts of the input. 

The label space is $\Sigma = \{0, 1\}^n$, and the reaction function for every $i$ is: 
\[\delta_i( (z_1, z_2, \ldots , z_n), x_i) = ((x_i, z_1, \ldots , z_{n-1}), f(z))\,.\] Denote by $\text{dist}(i,i')$ the distance from $i$ to $i'$, i.e., the length of the shortest directed path that starts at $i$ and ends at $i'$. Note that for every two nodes $i, j$, node $i$ sees $x_j$ at index $\text{dist}(j,i)$ of its incoming label, after $\text{dist}(j,i)$ time steps from initialization. Therefore, after $n$ time steps the labels converge and every node sees the input in a cyclic shift and computes the output correctly.

\end{proof}
\fi

\section{Proofs for Section~\ref{ssec:hardness}}\label{apdx:hardness}

\begin{theorem}
Let $A_n = (\Sigma,(\delta_1, \ldots, \delta_n))$ be a stateless protocol. Consider the following 2-party communication problem. Alice receives as input a description of $\delta_1$ and Bob receives $\delta_2$. They both have access to $\delta_3, \ldots, \delta_n $, and they need to decide whether $A_n$ is label $r$-stabilizing. 

For every $r$, deciding whether $A_n$ is label $r$-stabilizing requires $2^{\Omega(n)}$ bits of communication.    
\end{theorem}

\begin{proof}[Proof of Theorem~\ref{thm:cc}]

The proof is divided into two parts, the first for $r\leq 2^{n/2}$ and the second for $r\geq 2^{n/2}$. In both proofs we utilize the hypercube $Q_{n-2}$ for construction, and specifically, the \emph{snake-in-the-box} in $Q_{n-2}$.

\begin{definition}
An induced simple cycle in an hypercube is called a \emph{snake-in-the-box}.
\end{definition}

\begin{theorem}[Abbott and Katchalski (\cite{abbott1988snake})]
Let $s:\mathbb{N}\rightarrow\mathbb{N}$ be the function that maps $d$ to the length of the maximal snake in $Q_{d}$. For every $d\geq 8$:
\[
\lambda 2^d \leq s(d) \leq 2^{d-1}
\]
For some $\lambda\geq 0.3$. 
\end{theorem}

In both reductions, for $r\leq 2^{n/2}$ and for $r\geq 2^{n/2}$, we construct protocols over the clique $K_n$ with the label space $\Sigma = \{0, 1\}$.
We define reaction functions that maps the same outgoing label to all neighbors. Formally, for every $i$ and every $j, k\neq i$ $\ell_{(i, j)}=\ell_{(i, k)}$. Since the graph considered is the complete graph, we use the notation $\ell_i \in \{0, 1\}$ to refer to $i$'s outgoing labeling. Thus, \[\ell_{-i} = (\ell_1, \ldots, \ell_{i-1}, \ell_{i+1}, \ldots, \ell_{n})\].
We also redefine $\delta: \{0, 1\}^n\times 2^{[n]} \to \{0, 1\}^n$, to depict the transition in the outgoing labeling of all nodes: 
\[
(\ell_1^t, \ldots, \ell_n^t) = \delta(\ell_1^{t-1}, \ldots, \ell_n^{t-1}) 
\]

\begin{theorem}\label{thm:cc4smallRs}
Deciding if a stateless protocol $A$ is label $r$-stabilizing for $r\leq 2^{n/2}$ requires $2^{\Omega(n)}$ bits of communication.    
\end{theorem}
\begin{proof} 
To simplify presentation, we first prove for $r=1$, and then generalize the proof for every $r\leq 2^{n/2}$. 
	Let $S$ be a maximal snake in $Q_{n-2}$. Alice and Bob gets as input $x, y \in \{0, 1\}^{|S|}$, respectively. Their reaction functions are: 
\[
\delta_1(\ell_{-1}) = 
\begin{cases}
x_i & \text{If } (\ell_3, \ldots, \ell_n) = s_i\in S\\
1             & \text{Otherwise }
\end{cases}
\]

\[
\delta_2(\ell_{-2}) = 
\begin{cases}
y_i & \text{If } (\ell_3, \ldots, \ell_n) = s_i \in S\\
0             & \text{Otherwise }
\end{cases}
\] 

The reaction function $\delta_j$ of every $j\geq 3$:
\[
\delta_j(\ell_{-j}) = 
\begin{cases}
0 & \text{If } \ell_1\neq \ell_2\\
\phi_j((\ell_3, \ldots, \ell_{j-1}, \ell_{j+1}, \ldots, \ell_{n}))    & \text{Otherwise }
\end{cases}
\]
Note that in order to define $\phi_3, \ldots, \phi_n$ it is suffices to orient $Q_{n-2}$. 
Observe that there exists an edge $(v_i,v_j)\in Q_{n-2}$ such that $v_i,v_j\notin S$. As the total number of edges is $(n-2)2^{n-3}$, and the number of edges that have at least one end coincide with $S$ is $|S|+(n-4)|S| \leq (n-3)2^{n-3}$, there are at least $2^{n-3}$ edges that do not coincide with $S$. W.l.o.g we assume $v_i = 0^{n-2}$.
We orient $S$ so as to create a directed cycle (in one of the two possible ways); orient $0^{n-2}$ towards $v_j$; orient all other edges towards $S$.

\begin{claim}
If $x\neq y$ then the protocol is label $1$-stabilizing.
\end{claim}
\begin{proof}
Let $\ell^0$ be the initial labeling.
If  $\ell_1^0\neq \ell_2^0$, then after one round $(\ell_3,\ldots,\ell_n)=0^{n-2}$. After the second round, $\ell^2$ equals to either $(1,0,v_j)$ or $(1,0,0^{n-2})$ and in both cases $\delta(\ell^2, [n]) = (1,0,0^{n-2})$. Note that this is a stable labeling.

If $\ell_1^0 = \ell_2^0$, from our orientation and the fact that $x\neq y$, at some point we reach a labeling in which $\ell_1\neq l_2$, and as described above after three rounds the labeling converges to $(1,0,0^{n-2})$.  
\end{proof}

\begin{claim}\label{clm:xyeq}
If $x=y$ then the protocol is not label $1$-stabilizing.
\end{claim}
\begin{proof}
For every initial labeling of the form $\ell^0 = (\alpha,\alpha, s_i)$, since $x=y$ 
we have that for every $\ell^t = (\alpha, \alpha, s_{i(t)})$, $\delta(\ell^t, [n]) = (\alpha, \alpha, s_{i(t)+1 \text{ mod }|S| })$, thus, the protocol oscillates.
\end{proof}
It follows that deciding if the protocol is label $1$-stabilizing is equivalent to deciding if $x=y$. Since the communication complexity of $\text{EQ}_d$ is $\Omega(d)$ we get that we need at least $|S| = s(n-2) \geq \lambda 2^{n-2}$ bits of communication. 

To generalize the reduction for every $r\leq 2^{n/2}$ we consider the $Q_{n-4}$ cube, instead of $Q_{n-2}$. We define $S$ to be the largest snake in $Q_{n-4}$, and partition $S = (S_1,\ldots, S_k)$ so that each $S_i$ is of size $3r$, and we modify the reaction functions as follow:

\[
\delta_1(\ell_{-1}) = 
\begin{cases}
x_i & \text{If } (\ell_3,\ell_4) \neq (1,1) \text{ and } (\ell_5, \ldots, \ell_n) = s_j\in S_i\\
1             & \text{Otherwise }
\end{cases}
\]

\[
\delta_2(\ell_{-2}) = 
\begin{cases}
y_i & \text{If } (\ell_3,\ell_4) \neq (1,1) \text{ and } (\ell_5, \ldots, \ell_n)  = s_j \in S_i\\
0             & \text{Otherwise }
\end{cases}
\]

\[
\delta_3(\ell_{-3}) = \ell_4
\]

\[
\delta_4(\ell_{-4}) = 
\begin{cases}
1 & \text{If } \ell_3 = 1 \text{ or } \ell_1\neq \ell_2 \\
0             & \text{Otherwise }
\end{cases}
\] 

For every $j\geq 5$:
\[
\delta_j(\ell_{-j}) = 
\begin{cases}
\phi_j((\ell_5, \ldots, \ell_{j-1}, \ell_{j+1}, \ldots, \ell_{n})) & \text{If } (\ell_3,\ell_4) \neq (1,1) \\
0             & \text{Otherwise }
\end{cases}
\]
Where $\phi$ is defined the same way as before.

If $x = y$ then the same oscillation as in \ref{clm:xyeq} works here too, and since every $1$-fair scheduling is also $r$-fair scheduling we get that the protocol is not label $1$-stabilizing.

If $x \neq y$, note that if at some point, $\ell_3=\ell_4=1$ the labeling converge to $(1, 0, 1, 1, 0^{n-4})$. Also observe that if at some point $(\ell_5^t, \ldots, \ell_n^t)  = s_j \in S$ and $(\ell_3,\ell_4) \neq (1,1)$ then from the definition of $\phi$ and the reaction functions, at the next time step $(\ell_5^{t+1}, \ldots, \ell_n^{t+1})  = s_{j+1}$, since there is exactly one node in $\{5, \ldots, n\}$ that wants to change its labeling bit. Assume that there exists an oscillation. If at some point $(\ell_5^t, \ldots, \ell_n^t)  = s_j \in S$ then the labeling must cycling $S$. But as there exists $i$ s.t. $x_i\neq y_i$ there exist $3r$ consecutive time steps when the labeling changes through $S_i$. After the first $r$ time steps nodes $1,2$ must change their labeling to $(\ell_1,\ell_2)=(1,0)$. After the next $r$ time steps node $4$ must change its labeling to $\ell_4=1$ and after the next $r$ time steps, $\ell_3=1$ also. From that point, the labeling must converge, a contradiction. If, on the other hand, there is no time step for which $(\ell_5^t, \ldots, \ell_n^t)  = s_j \in S$, then again $(\ell_1,\ell_2)=(1,0)$ for $2r$ consecutive time steps, so the labeling must converge in this case, too. Thus if $f$ is the corresponding function of deciding whether a protocol is label $r$-stabilizing  
\[
CC(f) \geq |x| =\frac{|S|}{3r} \geq \frac{\lambda2^{n-4}}{3\cdot 2^{\frac{n}{2}}} = \frac{\lambda}{3}2^{\frac{n}{2}-4} = 2^{\Omega(n)}
\]
\end{proof}

\begin{theorem}
Deciding if a stateless protocol $A$ is label $r$-stabilizing for $r\geq 2^{n/2}$ requires $2^{\Omega(n)}$ bits of communication.   
\end{theorem}
\begin{proof}
We show a reduction from \textsc{Set-disjointness}. Let $S$ be a maximal snake in $Q_{n-2}$. Alice and Bob get as input $x,y\in \{0, 1\}^{q}$ where $q = r/2$ if $r\leq |S|$, and $q = |S|$ for $r\geq |S|$. These vectors are characteristic vectors, indicate which subsets out of the $2^q$ subsets of $[q]$ each party holds. We divide $S = (s_1, \ldots, s_{|S|})$  into consecutive paths of length $q$, $(S_1, \ldots, S_{k-1})$. If $q$ divides $|S|$ then we define $S_k = \emptyset$, and otherwise, $S_k$ is the remaining path of length $|S|\text{ mod }q$. We construct the mapping $I(j)  = 1 + (j-1)\text{ mod }q$ for every $j$ such that $s_j\in S_{i}$.

The reaction  functions, $\delta_1$ and $\delta_2$:
\[
\delta_1(\ell_{-1}) = 
\begin{cases}
x_{I(j)} & \text{If } \ell_2 = 0 \text{ and } (\ell_3, \ldots, \ell_n) = s_j\in S_i \\
0             & \text{Otherwise }
\end{cases}
\]

\[
\delta_2(\ell_{-2}) = 
\begin{cases}
y_{I(j)} & \text{if } \ell_1 = 0 \text{ and } (\ell_3, \ldots, \ell_n) = s_j \in S_i\\
0             & \text{otherwise }
\end{cases}
\] 

The reaction function $\delta_j$ of every $j\geq 3$:
\[
\delta_j(\ell_{-j}) = 
\begin{cases}
\phi_j((\ell_3, \ldots,\ell_{j-1}, \ell_{j+1}, \ldots, \ell_n) & \text{if } (\ell_1,\ell_2) = (1, 1)\\
0             & \text{otherwise }
\end{cases}
\]
Note that in order to define $\phi_3, \ldots, \phi_n$ it is suffices to orient $Q_{n-2}$. 
Observe that there exists an edge $(v_i,v_j)\in Q_{n-2}$ such that $v_i,v_j\notin S$. As the total number of edges is $(n-2)2^{n-3}$, and the number of edges that have at least one end coincide with $S$ is $|S|+(n-4)|S| \leq (n-3)2^{n-3}$, there are at least $2^{n-3}$ edges that do not coincide with $S$. W.l.o.g we assume $v_i = 0^{n-2}$.
We orient $S$ so as to create a directed cycle (in one of the two possible ways); orient $0^{n-2}$ towards $v_j$; orient all other edges towards $S$.

\begin{claim}\label{clm:ndistj}
If $E^A \cap E^B \neq \emptyset$ then the protocol is not label $r$-stabilizing.
\end{claim}
\begin{proof}
Assume $k\in [q]$ is the index corresponding to a subset in $E^A \cap E^B$, and we denote by $j_i$ the index $s_j\in S_i$ s.t. $I(j)=k$. We initialize the system with $(1, 1, s_1)$, and the following $r$-fair scheduling:
\begin{enumerate}
\item For $k-1$ times do:
\begin{enumerate}
\item Activate $1,2$ twice.
\item Activate nodes $2, \ldots, n-2$ for $q$ time-steps.
\end{enumerate}
\item Activate $2, \ldots, n-2$ for $q+|S_k|$ time-steps. 
\item Go back to step $1$.
\end{enumerate}
Note that the maximal time steps between activation is $q+|S_k|< 2q \leq r$, thus, the scheduling is $r$-fair. The system is in oscillation from the definition of the reaction functions.
\end{proof}

\begin{claim}\label{clm:distj}
If $E^A \cap E^B = \emptyset$ then the protocol is label $r$-stabilizing.
\end{claim}
\begin{proof}
Note that if at some point it is never the case that $(\ell_1,\ell_2) = (1, 1)$, then from that point, the system will converge to $0^{n}$. On the other hand, if $(\ell_1,\ell_2) = (1, 1)$ infinitely many times then: (1) they were activated simultaneously  since they pick $1$ only if the the other has label $0$ (2) $(\ell_3, \ldots, \ell_n)=s_j\in S_i$ for some $i$ (3) $y_{I(j)} = x_{I(j)}=1$, meaning that $E^A \cap E^B \neq \emptyset$.    

\end{proof}
Since the communication complexity of \textsc{Set-disjointness} is $\geq q$, if $f$ is the corresponding function of deciding whether a protocol is label $r$-stabilizing, where $|S| > r \geq 2^{n/2}$ then:
\[
CC(f) \geq q = r/2 \geq 2^{n/2 - 1} = 2^{\Omega(n)} 
\]
If $|S|\leq r$ then:
\[
CC(f) \geq q = |S| \geq \lambda 2^{n} = 2^{\Omega(n)}  
\]
\end{proof}
\end{proof}

\begin{theorem}
For every $r$, deciding whether a stateless protocol is label $r$-stabilizing is PSPACE-complete.
\end{theorem}
\begin{proof}[Proof of Theorem~\ref{thm:pspace}]

We show that deciding whether a protocol is label $r$-stabilizing is PSPACE complete.
Our proof of this theorem relies on an interesting ``stabilizing preserving reduction”. 
 We first (\thmref{thm:stateful}) show that the \textsc{String-oscillation} problem can be reduced to the problem of deciding whether a \emph{stateful protocol}, in which reaction functions may depend on their incoming and outgoing labels, is label $r$-stabilizing. Next ((\thmref{thm:stateful-less})), we show that we can construct a stateless protocol from a stateful protocol, without altering its stabilization properties.

\begin{theorem}\label{thm:stateful}
For every $r$, deciding whether a \emph{stateful protocol} is label $r$-stabilizing is PSPACE-complete.
\end{theorem}
\begin{proof}
We construct a reduction from the  \textsc{String-oscillation} problem defined as follow. Given a Boolean circuit that implements $g : \Gamma^m \rightarrow \Gamma \cup \{halt\}$, for some alphabet $\Gamma$, the objective is to determine whether there exists an initial string $T=(T_1,...,T_{m})\in \Gamma^m$ such that the following procedure does not halt.
\begin{enumerate} 
\item $i \leftarrow 1$
\item While $g(T)\neq$ \emph{halt} do
  \begin{enumerate}
  \item $T_i \leftarrow g(T)$
  \item $i \leftarrow 1+(i \text{ mod } m)$
  \end{enumerate}
\end{enumerate}
 
Given an instance of  \textsc{String-oscillation}, i.e., a function $g$, we construct a stateful protocol, in which reaction functions may depend on their incoming and outgoing labels over the clique $K_n$, where $n=m+1$. We define reaction functions that maps the same outgoing label to all neighbors. Formally, for every $i$ and every $j, k\neq i$ $\ell_{(i, j)}=\ell_{(i, k)}$. Since the graph considered is the complete graph, we use the notation $\ell_i \in \Sigma$ to refer to $i$'s outgoing labeling. Thus, we redefine $\delta_i: \Sigma^n\to  \Sigma$ so that,
\[
\ell_i^t = \delta_i(\ell_1^{t-1}, \ldots, \ell_n^{t-1})
\]
%\[\ell_{-i} = (\ell_1, \ldots, \ell_{i-1}, \ell_{i+1}, \ldots, \ell_{n})\].
We also redefine $\delta: \Sigma^n\times 2^{[n]} \to \Sigma^n$, to depict the transition in the outgoing labeling of all nodes: 
\[
(\ell_1^t, \ldots, \ell_n^t) = \delta(\ell_1^{t-1}, \ldots, \ell_n^{t-1}) 
\]

\begin{enumerate}
\item $\Sigma = [m]\times \Gamma\cup \{ \text{halt}\}$.
\item The reaction function of each node $i<n$:

\[
\delta_i((k_1, \alpha_1), \ldots, (k_n, \alpha_n), (j, \gamma)) = 
\begin{cases}
(1, \text{halt}) & \text{If } \gamma = \text{halt} \\
(1, \gamma) & \text{If } \gamma \neq \text{halt} \text{ and } j=i\\
(1, \alpha_i)             & \text{Otherwise }
\end{cases}
\]
\item The reaction function of node $n$: 
%\item For node $m+1$, assume $L^t_{m+1} = (j, \gamma)$:
\[
\delta_{m+1}((k_1, \alpha_1), \ldots, (k_n, \alpha_n), (j, \gamma)) = 
\begin{cases}
(1, \text{halt}) & \text{If } \gamma = \text{halt} \\
((j+1) \text{ mod } m, \text{ }g(\alpha_1, \ldots, \alpha_m))  & \text{If } \gamma \neq \text{halt} \text{ and } \alpha_j=\gamma\\
(j, \gamma)             & \text{Otherwise }
\end{cases}
\]
\end{enumerate}

\begin{claim}
If there exists $T'\in \Gamma^t$ s.t. the procedure doesn't terminate, then the protocol is not label $r$-stabilizing. 
\end{claim}
\begin{proof}
Starting with the initial outgoing labeling $\ell^0_i = (0,T'_i)$ for every $i<n$ and $\ell^0_{(m+1)} = (1,g(T'))$, every fair scheduling simulates the procedure, thus the protocol won't converge.
\end{proof}

\begin{claim}
If the protocol is not label $r$-stabilizing, then there exists $T'\in \Gamma^t$ s.t. the procedure doesn't terminate. 
\end{claim}
\begin{proof}
If the protocol is not label $r$-stabilizing, then there exists an initial labeling $\ell^0$ and $\sigma$ a $r$-fair scheduling so that the run oscillates. Let $t$ be a time step such that:
\begin{enumerate}
\item $t > 2rm$. 
\item $\ell^t_{(m+1)} = (m, \gamma)$ and $\ell^t_{m} = \gamma$.
\item $(m+1)\in \sigma(t)$.
\end{enumerate}
Since the scheduling is $r$-fair we know such $t$ exists. Denote $(1, \alpha_i) = \ell^t_{i}$, for every $i<m+1$, and note that $\ell^{t+1}_{i} = \ell^t_{i}$ and that $\ell^{t+1}_{(m+1)} = (1, g(\alpha_1, \ldots, \alpha_m))$.
Since any fair scheduling simulates the procedure, the run of the procedure with the initial string  $T = (\alpha_1, \ldots, \alpha_m)$ doesn't terminate.   
\end{proof}

\end{proof} 

\begin{theorem}\label{thm:stateful-less}
For every stateful protocol $A=(\Sigma, \delta)$ on $K_n$ there exists a stateless protocol $\bar{A}=(\bar{\Sigma}, \bar{\delta})$ on $K_{3n}$ such that $A$ is label $r$-stabilizing if and only if $\bar{A}$ is label $r$-stabilizing.
\end{theorem}

\paragraph{Topology Construction}
We map each node $i$ in $K_n$ to a triplet of nodes $\langle i_1, i_2, i_3 \rangle$ in $K_{3n}$ which we call \emph{metanode}.

Before proceeding to the protocol construction, we formulate some definitions needed for our formal statements.

\begin{definition}[Weakly and Strongly Consistent Metanode]
We say a metanode is \emph{weakly-consistent} if at least two of the nodes in the metanode have the same (output) label, which cannot be $\omega$. We say a metanode is \emph{strongly-consistent} if all nodes in the metanode have the same label.
\end{definition}

\begin{definition}[Weakly and Strongly Consistent Labeling]
A labeling is said to be \emph{weakly-consistent} if all metanodes are weakly-consistent, and all but at most one metanode are strongly-consistent.
A labeling is said to be \emph{strongly-consistent} if all metanodes are strongly-consistent.
\end{definition}

\begin{definition}[Corresponding Labeling]
For every weakly-consistent labeling $\bar{\ell}\in \Sigma^{3n}$ of $K_{3n}$ we call the unique labeling $\ell\in \Sigma^n$ \emph{ "$\bar{\ell}$'s corresponding labeling"} defined as:
\[
\ell_i = \text{MAJORITY}\{\bar{\ell}_{i_1}, \bar{\ell}_{i_2}, \bar{\ell}_{i_3}\} 
\].
\end{definition}

\begin{definition}[Consistent View]
We say that a node $i_j$ has a \emph{consistent-view} if (1) all matanodes besides its own metanode are strongly-consistent, and (2) the other two nodes in its own metanode have the same label, which cannot be $\omega$.
\end{definition}

\paragraph{Protocol Construction:}
We define the protocol $\bar{A}=(\bar{\Sigma}, \bar{\delta}_1, \ldots, \bar{\delta}_{3n})$ on $K_{3n}$ as follow:

\begin{enumerate}
\item $\bar{\Sigma} = \Sigma \cup \{\omega\}$.
\item The reaction function of node $i_j$ defined as follow:
\[
\bar{\delta}_{i_j}(\bar{\ell}_{-i_j}) = 
\begin{cases}
\omega, & \text{if $i_j$'s view is not  consistent}\\
\omega,           & \text{otherwise, if $L$ is a stabilizing labeling}\\
\delta_i(\ell), &  \text{otherwise}

\end{cases}
\]
Where $\ell$ is $\bar{\ell}$ corresponding labeling.
Note that the reaction function is well define since if $i_j$'s view is consistent it must be that $\bar{\ell}$ is weakly-consistent. 
\end{enumerate}

\begin{claim}
If $A$ is not label-stabilizing on $K_n$ then $\bar{A}$ is not label-stabilizing on $K_{3n}$.
\end{claim}
\begin{proof}
Let $\ell^0$ and $\sigma$ be an initial labeling and a fair scheduling for which $A$ oscillates. Then $\bar{A}$ oscillates on the fair scheduling, $\bar{\sigma}(t) = \bigcup_{i\in \sigma(t)} \{i_1, i_2, i_3\}$, 
and the initial labeling, $\bar{\ell}^0_{i_j} = \ell^0_i$.
\end{proof}

\begin{lemma}\label{lem:wc}
Let $\bar{\sigma}$ be a fair scheduling and $\bar{\ell}^0$ the initial labeling, such that each labeling in the labeling sequence, $\bar{\ell}^0, \bar{\ell}^1, \ldots$, is weakly-consistent.
Then, 
\begin{enumerate}
\item If $\bar{\ell}^k$ is weakly-consistent, and $i_j$ was activated, then $\delta_i(\ell^k) = \ell^k_i$.
\item If $\bar{\ell}^k$ is strongly-consistent and $P$ is the set of nodes who alter their label, then either $P$ is empty, or $|P|\geq 2$ and in this case, it contains complete metanodes and at most one partial metanode.
\end{enumerate}
\end{lemma}

\begin{proof}
\begin{enumerate}
\item Assume $\bar{\ell}^k$ is weakly-consistent, and let $i_j$ be the node that differ from the other nodes in $i$'s metanode. If $i_j$ is not activated or if the number of nodes that are activated is greater than two, then $\bar{\ell}^{k+1}$ cannot be weakly-consistent. Thus, either only $i_j$ was activated or $i_j$ and another node, $u_c$.

Note that $i_j$ has a consistent view, therefore $\bar{\ell}^{k+1}_{i_j} = \delta_i(\ell^k)$. Assume $\delta_i(\ell^k)\neq \ell^k_i$. If $i_j$ wasn't the only node that was activated then, $\bar{\ell}^{k+1}_{u_c} = \omega$, and as $i_j$ still differ from the nodes in $i$s metanode, we get that $\bar{\ell}^{k+1}$ is not weakly-consistent. If $i_j$ was the only node that was activated then, since the scheduling is fair after finite number of steps, other node must be activated, thus, we again have a non weakly-consistent labeling. Therefore, $\delta_i(\ell^k) = \ell^k_i$. If $i_j$ was not the only node that was activated then observe that the corresponding label of $\bar{\ell}^{k+1}$ hasn't change, that is $\ell^{k+1}=\ell^{k}$ and that $u_c$ is the node that differ from the nodes in its metanode in $\bar{\ell}^{k+1}$, so $\delta_u(\ell^k) = \ell^k_u$. 

\item Assume $\bar{\ell}^k$ is strongly-connected and let $P$ be the set of nodes who alter their label, i.e., $P = \{ u_c : u_c\in \bar{\sigma}(k) \text{ and } \bar{\ell}^k_{u_c}\neq \bar{\ell}^{k+1}_{u_c} \}$. If $P = \{i_j\}$ for some $i_j$, then note that $\bar{\ell}^{k+1}$ is weakly-connected and that $\delta_i(\ell^{k})\neq \ell^{k}_i$, which means that $\delta_i(\ell^{k+1})\neq \ell^{k+1}_i$, and from previous section this implies that $\bar{\ell}^{k+2}$ is not weakly-consistent. Thus, either $P=\emptyset$ or $|P|\geq 2$. 
Observe that if $P$ contains two non complete metanodes then $\bar{\ell}^{k+1}$ is not weakly-consistent, thus $P$ has at most one incomplete metanode.

\end{enumerate}
\end{proof}

\begin{claim}
If $A$ is label-stabilizing on $K_n$ then $\bar{A}$ is label-stabilizing on $K_{3n}$.
\end{claim}
\begin{proof}\label{prf:kn3kn}
First observe that if at some point the labeling in not weakly-consistent, or if the corresponding labeling is a stable, then from that point, for any fair scheduling the labeling will converge to $\omega^{3n}$. Assume that there exist an initial labeling $\bar{\ell}^0$ and a fair scheduling $\bar{\sigma}$ for which $\bar{A}$ doesn't converge. Then it must be that every labeling in $\bar{\ell}^0, \bar{\ell}^1, \ldots$ is weakly-consistent and its corresponding labeling is not stable. Consider the subsequence of strongly-consistent labelings $\{\bar{\ell}^{n_k}\}$ and let $P^k = \{ u : u\in \bar{\sigma}(n_k) \text{ and } \bar{\ell}^{n_k}_{u}\neq \bar{\ell}^{n_k+1}_{u} \}$.
We define the following scheduling, $\sigma$:

\[
\sigma(k) = S_1 \cup S_2 \cup S_3
\]
Where:
\begin{enumerate}
\item $S_1$ is the set of nodes who don't change their label in the corresponding labeling $L^k$. Formally, \[
S_1 = \{i : i_j\in \bar{\sigma}(n_k)\setminus P^k\}
\]
\item $S_2$ is the set of nodes that change their label and at least two of the nodes in their metanodes were activated. Formally, 
\[
S_2 = \{i : |\{i_1, i_2, i_3\}\cap P^k| \geq 2\}
\]
\item $S_3$ is the set of nodes that were activated in all the weakly but not strongly consistent labelings between the previous strongly-consistent labeling and the current labeling. Formally, 
\[
S_3 = \{i : i_j \in \bar{\sigma}(t) \text{ for every } n_{k-1} < t < n_k\}
\]
\end{enumerate}
We claim that $\sigma$ is a fair scheduling. To see this, observe that from lemma \ref{lem:wc} for every $t$ such that $\bar{\ell}^t$ is not strongly-consistent $\bar{\sigma}(t)\subseteq \sigma(k)$ where $t\leq n_k$ is the next strongly-consistent labeling, and that if $\bar{\ell}^t = \bar{\ell}^{n_k}$ is strongly-consistent then $i_j\in \bar{\sigma}(n_k)$ implies $i\in \sigma(k)$ unless only $i_j$ was activated among $i$'s metanode. Since $\bar{\ell}^{n_k +1}$ is not strongly-consistent, from the proof of lemma \ref{lem:wc} we know that $i_j\in \bar{\sigma}(n_k+1)$, thus, $i\in \sigma(k+1)$. So to proof that $\sigma$ is fair, we need to show that the subsequence $\{\bar{\ell}^{n_k}\}$ is infinite, but from the proof of lemma \ref{lem:wc}, is it was finite then $\bar{\sigma}$ is not fair. 

Finally, we claim that $A$ oscillates on $\ell^0$ and $\sigma$, defined above. Assume that $A$ converges. Note that this implies that from some point there are no weakly but not strongly consistent labelings, since the labeling $A$ converges to has a strongly-consistent corresponding labeling, and from lemma \ref{lem:wc}, the transition from strongly-consistent to non strongly-consistent labeling implies a change in the label of at least one metanode. Since from some point there are no weakly-consistent labelings, and from the definition of $\sigma$, $i$ changes its label if and only if $i$'s metanode changes its label, thus, $\bar{A}$ converges, a contradiction.

\end{proof}

\begin{theorem}
If $\bar{A}$ is not $r$-stabilizing then $A$ is not $r$-stabilizing. 
\end{theorem}
\begin{proof}
Following the proof \ref{prf:kn3kn}, note that every scheduling for which $\bar{A}$ oscillates, has a corresponding scheduling with only strongly-consistent labelings, that oscillates too, and since such scheduling can be translated directly to a scheduling for $A$, the statement holds.
\end{proof}
\begin{theorem}
If $A$ is not $r$-stabilizing then $\bar{A}$ is not $r$-stabilizing. 
\end{theorem}
\begin{proof}
Every scheduling in $A$ can be translated to a scheduling in $\bar{A}$ by activating only complete metanodes. 
\end{proof}
 
\end{proof}

\section{Proofs for Section~\ref{sec:os}}\label{apdx:os}

\begin{theorem}%\label{thm:lpoly}
	$\textup{OS}^u_{\log}\equiv \lpoly$.
\end{theorem}

\begin{proof}[Proof of Theorem \ref{thm:lpoly}]
We prove this theorem by showing that any family of protocols with logarithmic label complexity on the unidirectional ring can be simulated by a logspace Turing machine with polynomial advice, and vice versa.

\vspace{1em}
\noindent[$\textup{OS}^u_{\log}\subseteq\lpoly$]
%\begin{proof}
%[$\textup{OS}^u_{\log}\subseteq\lpoly$]\\
Let $\mathcal{L}\in\textup{OS}^u_{\log}$. Then there is some family $A_1,A_2,\ldots$ of protocols with logarithmic label complexity that decide $\mathcal{L}$, meaning that each $A_n$ computes $f_n$, $\mathcal{L}$'s indicator function on $\{0,1\}^n$. We describe a logspace Turing machine $M$ and a polynomial-length advice function $a:\bN\to\{0,1\}^*$ such that, for every $n\in\bN$ and $x\in\{0,1\}^n$, the machine accepts the input $(x,a(n))$ if $x\in\mathcal{L}$ and rejects it otherwise. For each $n$, the advice string $a(n)$ will contain the truth tables of the reaction functions $\delta_1,\ldots,\delta_n$. Since $|\Sigma|$ is polynomial in $n$, the description of the $n$ functions $\delta_i:\Sigma\times\{0,1\}\to\Sigma\times\{0,1\}$ is also of length polynomial in $n$.

    To simulate the protocol $A_n$, the machine will maintain on its work tape a private output variable $y_j$ of length $1$, an incoming label variable of length $\log(|\Sigma|)$, and two counters: $j$, of length $\log(n)$, and $t$, of length $\log(n)+\log(|\Sigma|)$. Let $\ell_0$ be an arbitrary label from $\Sigma$. The machine executes the following procedure.
\begin{myalgo}{Simulation of protocol $A_n$}
	    	$y,t\leftarrow 0$\\
	    	$j\leftarrow 1$\\
	    	$\ell\leftarrow \ell_0$\\
	    	While $t<n|\Sigma|$:\\
	    	\qquad $(\ell,y_j) \leftarrow \delta_j(\ell,x_j)$\\
	    	\qquad $j\leftarrow j\pmod n + 1 $\\        
	    	\qquad $t\leftarrow t+1$\\
	    	Output $y$
\end{myalgo}
After the $t$\textsuperscript{th} iteration of this loop, $\ell$ and $y$ will reflect the outgoing label and output of node $t\bmod n$ after the protocol $A_n$ has been run for $t$ time steps on input $x$ and initial labeling $(\ell_0,\ldots,\ell_0)$. 

We claim %Proposition~\ref{prop:rbound} tells us that 
each node's output converges to $f_n(x)$ within $n|\Sigma|$ time steps, so at the end of this procedure we have $y=f_n(x)$. Thus $M$ decides $\mathcal{L}$ using logarithmic space and polynomial advice, i.e., $\mathcal{L}\in\lpoly$. Our claim follows from the following lemma, that establishes the relation between the label complexity and round complexity in unidirectional rings. %\\

\begin{lemma}\label{prop:rbound}
Let $G=([n], E)$ be the unidirectional ring and $A_n=(\Sigma, \delta)$, a stateless protocol. Then,

\begin{enumerate} 
\item $R_n \leq n|\Sigma| = n2^{L_n}$
\item There exists a protocol with $R_n = n(|\Sigma|-1)$. 
\end{enumerate}

\end{lemma}
\begin{proof} %[Proof of Proposition \ref{prop:rbound}]
\begin{enumerate}
\item  Observe that if $\ell_{-1}^{n} = \ell_{-1}^{0}$ then $\ell_{-2}^{n+1} = \ell_{-2}^{1}$. Thus, at time $2n$, we have that $\ell_{-1}^{2n} = \ell_{-1}^{n}$. From the same reasoning, if $\ell_{-1}^{n} \neq \ell_{-1}^{0}$ but there exists $k\in\mathbb{N}$ for which $\ell_{-1}^{kn} = \ell_{-1}^{0}$ then for every $j\in \mathbb{N}$, $\ell_{-1}^{jkn} = \ell_{-1}^{0}$. 

	Since $R_n$ is the round complexity of $A_n$, we know that for some input assignment and some labeling, the convergence time is $R_n$. We initialize the system with this configuration, and look at the system at time steps $0, n, 2n, \ldots, |\Sigma|n$.
We define for every $k\in \{0, 1, \ldots, |\Sigma|\}$ the set:
\[B_{kn} = \{ i : \ell_{-i}^{kn}\notin\{ \ell_{-i}^{0}, \ell_{-i}^{n}, \ldots, \ell_{-i}^{(k-1)n} \}  \}.\] 
From the observation above, if a node $j\notin B_{kn}$ then $j\notin B_{rn}$ for any $r>k$. We claim that if $B_{kn}$ is empty then the system is in oscillation. Assume that at time $kn$, $B_{kn} = \emptyset$. It means that for every node $i$ there is $r_i\in \mathbb{N}$ and $(\ell_{-i}^0, \ldots, \ell_{-i}^{r_i-1})\in \Sigma^{r_i}$ so that i's incoming label at time $(k+j)n$ is $\ell_{-i}^{j \text{ mod } r_i}$, for every $j$. 
Denote by $\ell^{kn}$, the labeling at time $kn$, and we define $r=\Pi_i r_i$. Note that $\ell^{(k+r)n} = \ell^{kn}$ hence, the system is oscillating,
and since $A_n$ is output stabilizing the system has reached (output) convergence. Finally, note that $B_{|\Sigma|n}$ must be empty since for every $i$, $r_i\leq|\Sigma|$, thus, $R_n\leq n|\Sigma|$.

\item  We now show an algorithm $A_n$ with $R_n=n(|\Sigma|-1)$. 
\begin{itemize}
\item $\Sigma = [q]$.
\item The reaction function of node $1$ is:
\[\delta_1(\ell_{-1},x_1)=\left\{
	 	\begin{array}{ll}
		 	(q-1, 1)&\ \text{if}\ \ell_{-1}=q-1\\
		 	(\ell_{-1}+1, 0)&\ \text{otherwise}\,.
		\end{array}\right
	.\]
\item The reaction function of node $i\neq 1$ is:
\[\delta_i(\ell_{-i},x_1)=\left\{
	 	\begin{array}{ll}
		 	(q-1, 1)&\ \text{if}\ \ell_{-i}=q-1\\
		 	(\ell_{-i}, 0)&\ \text{otherwise}\,.
		\end{array}\right
	.\]

\end{itemize}
We initialize the system with an arbitrary input assignment (as the reaction function only depends on the incoming label) and initialize the labeling with $\ell_{(1,2)}=0$ and for every other label, $\ell=0$. It is not hard to see that every $n$ time steps the outgoing label of node $1$ increases by $1$ and that only at time $t_{(q-1)n}$ all outputs stabilize. Thus, $R_n=n(q-1) = n(|\Sigma|-1)$.
\end{enumerate}
\end{proof}

    \noindent[$\lpoly\subseteq\textup{OS}^u_{\log}$] Now let $\mathcal{L}\in\lpoly$. Then there exist functions $s:\bN\to\bN$, $a:\bN\to\{0,1\}^*$ and a Turing machine $M$ that decides $\mathcal{L}$ such that $s(n)=O(\log(n))$, and $M$ accepts $(x,a(n))$ if and only if $x\in\mathcal{L}$.
	
	Fix $n$ and the corresponding advice string $a(n)$, and let $Z=Q\times\{0,1,\textvisiblespace\}^{s(n)}\times[s(n)]\times[n]$, the set of all possible configurations of $M$'s work tape and input tape head on inputs of size $n$. For each $(q,w,k,j)\in Z$, $q$ is the machine state, $w$ is the content of the work tape, $k$ is the position of the work tape head, and $j$ is the position of the input tape head. Define the initial configuration $z_0=(q_0,\textvisiblespace^{s(n)},1,1)$, and define a function $F:Z\to\{0,1\}$ by $F(q,w,k,j)=1$ if and only if $q$ is an accepting state. The transition function of $M$ induces a partial function $\pi:Z\times\{0,1\}\to Z$ such that if $M$ is in configuration $z$ and reads input bit $b$, then its next configuration will be $\pi(z,b)$.
	 
	 We now describe the protocol used to simulate $M$ on inputs of size $n$. The label space is $\Sigma= Z\times\{0,1\}\times[|Z|]\times\{0,1\}$.  The first two coordinates of a label $(z,b,c,o)\in\Sigma$ correspond to the domain of the function $\pi$. The third, $c$, is a counter used to periodically reset the initialization in our simulation of $M$, and the last, $o$, communicates the current output.
	 
	Informally, in our protocol node $1$ will run $n$ simulations of $M$ in parallel, and it will be responsible for initializing and updating configurations in these simulations, while the other nodes will only answer queries about their input bits and forward information. Formally, node $1$ has reaction function
	 \[\delta_1((z,b,c,o),x_1)=\left\{
	 	\begin{array}{ll}
		 	((\pi(z,b),x_1,c+1,o),o)&\ \text{if}\ c<|Z|\\
		 	((z_0,x_1,0,F(z)),F(z))&\ \text{otherwise}\,,
		\end{array}\right
	.\]
	and each node $i=2,\ldots,n$ has reaction function
	\[\delta_i((z,b,c,o),x_i)=\left\{
		\begin{array}{ll}
			((z,x_i,c,o),o)&\ \text{if the input tape head position in $z$ is $i$}\\
			((z,b,c,o),o)&\ \text{otherwise}\,.
		\end{array}\right.
	\]
	Let $r\in[n]$, and consider the sequence of outgoing labels of node $1$ at time steps $n+r,2n+r,3n+r,\ldots$ when this protocol is run on input $x\in\{0,1\}^n$. Regardless of the initial labeling, the counter $c$ in this label sequence will eventually reach $|Z|$ at time $qn+r$, for some $q\in\bN$. After that, the label sequence will reflect the first $|Z|$ configurations of $M$ when it is properly initialized and run on $x$ with advice $a(n)$. Thus, at time $(q+|Z|)n+r$, node 1 will output 1 if and only if the simulation of $M$ is in an accepting configuration at time $|Z|-1$, i.e., if and only if $x\in\mathcal{L}$. Node 1 will then re-initialize the simulation, but for every $k\geq q+|Z|$, it will always have the correct output in time step $kn+r$. Additionally, because node 1 passes on its output as part of its label, each node $i\in[n]$ will have the correct output at time step $kn+r+i-1$. This argument applies for every $r\in[n]$, so every node's output will eventually converge to the correct value. We conclude that $\mathcal{L}\in \text{OS}_{\log}^u$.
\end{proof}

%\section{Proofs for Section~\ref{sec:os}}\label{apdx:os}

\begin{theorem}%\label{thm:ppoly}
	$\widetilde{\textup{OS}}^b_{\log} \equiv \ppoly$.
\end{theorem}

\begin{proof}[Proof of Theorem \ref{thm:ppoly}] 

The first part of the proof shows that any output-stabilizing protocol with polynomial round complexity on the bidirectional ring, can be emulated by a polynomial sized Boolean circuit. The key observation here, is that any function, $f:\{0, 1\}^n\rightarrow f:\{0, 1\}^m$, can be implemented via a Boolean circuit whose size is exponential in $n$. We use this fact to build polynomial (in $|V|$) sized circuits that implement the nodes reaction functions.

\vspace{1em}
\begin{proof}[$\logBOS \subseteq \ppoly$] Let $\{A_n\}_{n=1}^{\infty}$, 
  a protocols family computing $\mathcal{L}$. We construct a
  family of polynomial size Boolean circuits $\{C_n\}_{n=1}^{\infty}$
  that satisfy $C_n(x) = 1$ iff $A_n(x)=1$. 
  Circuit construction:

\begin{enumerate}
  
\item We use small circuits, $\{C_{\delta_i}\}_{i=1}^{n}$, serve as building blocks, that compute the reaction function of each
  node. We use the fact that any Boolean function,  $g:\{0,1\}^N\rightarrow
  \{0,1\}^M$, can be computed by a Boolean circuit of size
  $MN2^N$. As the label size is logarithmic, and the in-degree and out-degree is $2$, each circuit, $C_{\delta_i}$, is of size polynomial in $n$. 

\item We use a constant circuit, $C_{0}$, that outputs $|E|L_n = 2nL_n$ zeros, simulating the initial labeling. 

\item Let $T$ be the maximal time, over all inputs, it takes $A_n$ to converge when the labels are initialized to zeros. The circuit is composed of $T$ layers, in each layer $j$, the $n$ circuits that compute the $n$ reaction functions, $\{C_{\delta_i}^j\}_{i=1}^{n}$. The first layer is connected to the initial labeling circuit, $C_0$. Any two consecutive layers, $0<j,j+1$ are connected according to the ring topology.

\item The input of $C_n$ will be $x_1, ... , x_n$ and each $x_i$ will be wired to all circuits $C_{\delta_i}^1, \ldots, C_{\delta_i}^T$. 
 
\item Recall that at convergence, 
  the private output of every node is $f_{\mathcal{L}}(x)$, thus, it 
  is suffices to define the output gate to be the private output of node $1$, $y_1$, at layer $T$. 
  
\item As $T \leq R_n$ is polynomial in $n$, the size of the circuit is also polynomial in $n$.

\end{enumerate}
\end{proof}

For the other direction, we use two protocols that assist us simulating a Boolean circuit. The first is a counter protocol, that we use to carry out a global clock in a form of a counter that counts up to some value, $D$, and repeats it indefinitely. The second, shows how a communication between one node and its neighbor can serve as a method to retain memory throughout the execution. The specification and proof of the $D$-counter protocol appears in Claim~\ref{clm:counter}.

The second assisting protocol shows how to retain memory via a communication, and is based on the counter protocol introduced above. We demonstrate this with an example. Consider a bidirectional ring of size $n=3$, with input vector $(x_1, x_2, x_3)$ and assume that node $2$ wants to save at some point the value $g(x_1, x_2)$, where $g$ is some logic gate. To achieve that, it communicates with node $3$ as follow. First, it sends $3$ a message with $g(x_1, x_2)$ (assumed at this point node $1$ send it its private input) at the $j^{th}$ coordinate, then, $3$ repeatedly updates $\ell_{(3,2)}^{t+1}[j] = \ell_{(2,3)}^t[j]$ and in turn, node $2$ updates,  $\ell_{(2,3)}^{t+1}[j] = \ell_{(3,2)}^t[j]$. From that point, node $2$ has ``saved''  $g(x_1, x_2)$ it computed earlier. 

The main challenge here, is how to do this in a self-stabilization stateless manner. Having a consistent counter makes things easier, and we can define $\delta_2$ to send $\ell_{(2,3)}[j] = g(x_1, x_2)$ if the counter value is either $x$ or $x+1$ ($x\in [D]$), and at other counter values to send $\ell_{(2,3)}^{t+1}[j] = \ell_{(3,2)}^t[j]$, and accordingly, define $\delta_3$ to send  $\ell_{(3,2)}^{t+1}[j] = \ell_{(2,3)}^t[j]$ for any counter value. The reason we need node $2$ to send the value $g(x_1, x_2)$ in two consecutive time steps is that only then, the $j^{th}$ coordinate stabilizes in both $\ell_{(2,3)}$ and $\ell_{(3,2)}$. 

Equipped with the counter protocol and a method for keeping memory throughout a run, we are ready to describe the circuit simulation over an odd sized bidirectional ring.

\vspace{1em}
\begin{proof}[$\ppoly \subseteq \logBOS$]
Let $C_n$ be a Boolean circuit of polynomial size, and fan-in $2$. Let $g_1, g_2, \ldots, g_{|C_n|}$ be the circuit gates, in a topological order.
We build a bidirectional $N$-ring as follow:

\begin{enumerate}

\item The size of the ring is $N=2|C_n|+n$ if $n$ is odd and  $N=2|C_n|+n+1$ if $n$
is even, as we need $N$ to be odd. 

\item The first $n$ nodes are the circuit input.

\item Additionally, for every gate $g_j$ we
have two nodes, $n+2j-1$ and $n+2j$. Node $n+2j-1$ calculates the gate value, and node $n+2j$ ``remembers'' the value, once calculated. 

\item The first label fields, denoted by $F_D$, are used to implement a $D$ counter.  Let $g_j$ be some gate, we define $(1(j), 2(j))$ to be the two vertices that correspond to $g_j$ inputs in $C_n$. For example, if $g_j$ is an AND gate of gate $g_2$ and input bit $x_5$, then, $(1(j), 2(j)) = (n+3, 5)$. We consider the clockwise orientation and define $d_j=\text{dist}(\text{min}\{ 1(j), 2(j) \}, n+2j-1)$, the length of the clockwise path from $g_j$'s minimal input node to $n+2j-1$. The counter value, $D$, is defined as $D =\sum_{j=1}^{|C_n|} (d_j+1) = |C_n|+\sum_{j=1}^{|C_n|} d_j$.

\item In addition to the counter fields we use four more bit fields, and a label is then $(F_D, i_1, i_2, v, o)$, where $i_1, i_2, v, o\in \{0, 1\}$.

\end{enumerate}
We denote $t_{j} = \sum_{k=1}^{j} (d_k+1)$.
We partition the counter values to $|C_n|$ intervals, where the $j$ interval is: 
\[
I_j = \{ t_{j-1} ,\ldots, t_j - 1 \}.
\]

The calculation of gate $g_j$ is during the $I_j$ interval, whose size is $d_j+1$, and it is performed as follows. 

All messages propagate in a clockwise direction from $i$ to $i+1$.
If the counter is either $t_{j-1}$ or $t_{j-1}+1$ then node $k = \text{min}\{ 1(j), 2(j) \}$ changes the $i_1$ field of $\ell_{(k, k+1)}$ to its value bit. If $k$ is an input node then the value bit is its input bit, $x_k$, and if $k$ is a gate node then it means it has calculated the value before (as the order or execution is topological) and it remembers it via the $v$ bit field from its ``memory'' node, $k+1$.

If the counter value is either $t_{j-1}+\text{dist}(1(j), 2(j))$ then the second input node, $k'$, updates $i_2$ in $\ell_{(k', k'+1)}$ according to its type (input node or a gate).
If the counter value is either $t_j-2$ or $t_j-1$ then, node $n+2j-1$ updates the $v$ field on $\ell_{(n+2j-1, n+2j)}$ with $v=g_j(i_1, i_2)$.
For any counter value, node $n+2j$ updates the $v$ field on $\ell_{(n+2j, n+2j-1)}$ from the $v$ field on $\ell_{(n+2j-1, n+2j)}$. From the method for keeping memory we saw, we are guaranteed that from this point and on, $v = g_j(i_1, i_2)$.

Finally, node $n+2|C_n|$ which sees at the $v$ field of $\ell_{(n+2|C_n|-1, n+2|C_n|)}$ the value of gate $g_{|C_n|}$, which is also the value of the circuit output, updates the $o$ field in both of its outgoing labels to $v$, for any counter value. The rest of the nodes propagate this value, thus, at convergence, the $o$ value converges to the circuit output, and every node in the system can update its private output to this value, as required. 

The round complexity is $R_N = 4N+D+N$, since  and the label complexity is $L_n=3\log(D)+6$. As $D\leq N+N^2$, and $N$ is polynomial in $n$ (since  the circuit size is polynomial in $n$) we get that the round complexity is polynomial in $n$ and the label complexity is logarithmic in $n$.
\end{proof}

\end{proof}

%\section{Proofs for Section~\ref{sec:hardFunc}}\label{apdx:hardFunc}
\iffalse
\begin{theorem}%\label{thm:hardFunc}
Let $\{G_n\}_{n=1}^{\infty}$ be graph family, so that the maximal degree  of $G_n$ is  constant, i.e., there is some $k\in \mathbb{N}$ so that for every $n$, $\Delta(G_n) = k$. Then for every $n>8$, there exists a function $f:\{0,
1\}^n\rightarrow \{0, 1\}$ that cannot be computed by a stateless protocol on $G_n$ with $L_n < n/(4k)$.
\end{theorem}

\begin{proof}[Proof of Theorem \ref{thm:hardFunc}]
The number of possible protocols over the label space $\Sigma$ is at most $(2|\Sigma|^k)^{2n|\Sigma|^k}$.
Applying the same protocol for two different Boolean functions on an
inconsistent input must result in an incorrect behavior. It follows that the numbers of protocols is at least the number of Boolean functions, $2^{2^n}$:
\[
(2|\Sigma|^k)^{2n|\Sigma|^k} \geq 2^{2^{n}}
\] 
\[
2n|\Sigma|^k \cdot \log(2|\Sigma|^k) \geq 2^n
\]
\[
\log(n) +\log(2|\Sigma|^k)+\log( \log(2|\Sigma|^k) )\geq n
\]
\[
\log(n)+2k\log(|\Sigma|) \geq n
\]
\[
2k\log(|\Sigma|) \geq n-\log(n)
\]
\[
\log(|\Sigma|) \geq \frac{1}{2k} (n-\log(n)) \geq \frac{1}{2k} (n-\frac{n}{2})
\]    
\[
L_n \geq \frac{n}{4k}\;.
\]

\end{proof}
\fi

\section{Proofs for Section~\ref{sec:ls}}\label{apdx:lsbound}

\begin{theorem}
%\label{thm:lowerbound}
	Let $m\in\bN$, let $f:\{0, 1\}^{n} \rightarrow \{0, 1\}$ be a function, and let $G=([n],E)$ be a directed graph. Define $C=\{(i,j)\in E:i\leq m<j\}$ and $D=\{(i,j)\in E:j\leq m<i\}$,
	the sets of edges out of and into the node subset $[m]$. Suppose $f$ has a fooling set $S\subseteq\{0,1\}^{m}\times\{0,1\}^{n-m}$ such that, for every $(x,y),(x^\prime,y^\prime)\in S$,
	\begin{itemize}
		\item if $(i,j)\in C$, then $x_{i} = x^\prime_{i}$, and
		\item if $(i,j)\in D$, then $y_i = y^\prime_i$.
	\end{itemize}
	Then every label-stabilizing protocol on $G$ that computes $f$ has label complexity at least $\frac{\log(|S|)}{|C|+|D|}$.
\end{theorem}

\begin{proof}[Proof of Theorem~\ref{thm:lowerbound}]
	Assume that $m$, $f$, $G$, $C$, $D$, and $S$ are as in the theorem statement. Let $A_n=(\Sigma,\delta)$ be a label-stabilizing protocol that computes $f$ on $G$.

	Let $(x,y),(x',y')\in S$ be distinct inputs, and assume without loss of generality that $f(x,y')\neq b$. Let $\ell^1, \ell^2\in E^{\Sigma}$ be any two (possibly identical) labelings. Let $\hat{\ell}^1$ and $\hat{\ell}^2$ be the eventual labelings when $A_n$ is run on input $(x,y)$ and initial labeling $\ell^1$, or input $(x',y')$ and initial labeling $\ell^2$, respectively. %the as follows. For each $(i,j)\in E$, we define $L_1(i,j)$ and $L'_1(i,j)$ to be the eventual label of $(i,j)$ when $A_n$ is run on input $(x,y)$ and initial labeling $L_0$, or input $(x',y')$ and initial labeling $L'_0$, respectively.
	
	Now assume for contradiction that, for every $(i,j)\in C\cup D$, we have $\hat{\ell}^1_{(i,j)}=\hat{\ell}^2_{(i,j)}$. Run $A_n$ on input $(x,y')$ with initial labeling $\ell^0$ given by
	\[\ell^0_{(i,j)}=\left\{\begin{array}{ll}\hat{\ell}^1_{(i,j)}&\ \text{if}\ i\in[m]\\\hat{\ell}^2_{(i,j)}&\ \text{otherwise}\,\end{array}\right.\]
	Consider any node $j\in[m]$. Every incoming edge $(i,j)$ is labeled by $\hat{\ell}^1_{(i,j)}$ if $i\in[m]$ and by $\hat{\ell}^2_{(i,j)}$ otherwise. But if $i\not\in[m]$, then $(i,j)\in D$, so $\hat{\ell}^2_{(i,j)}=\hat{\ell}^1_{(i,j)}$. Thus, the inputs to $j$'s reaction function $\delta_j$ are exactly what they would be if the global input and labeling were $(x,y)$ and $\hat{\ell}^1$, which means that $\delta_j$ is at a fixed point and will output $f(x,y)$. Similarly, each node $j'\in [n]\setminus[m]$ receives the same incoming labels and input as if the global input and labeling were $(x',y')$ and $\hat{\ell}^2$, and is therefore at a fixed point and outputting $f(x',y')$. It follows that the protocol is at a global fixed point in which some nodes output $f(x,y)\neq f(x,y')$, contradicting the assumption that $A_n$ computes $f$.
	
	We conclude that there is some $(i,j)\in C\cup D$ such that $\hat{\ell}^1_{(i,j)} \neq \hat{\ell}^2_{(i,j)}$, and this must be true for every choice of $(x,y),(x',y')\in S$ and $\ell^1,\ell^2$. It follows that there is a one-to-one mapping from inputs in $S$ to labelings of the edges in $C\cup D$, i.e., $|S|\leq |\Sigma|^{|C\cup D|}$, so $\log(|\Sigma|)\geq \log(|S|)/(|C|+|D|)$.

\end{proof}

\begin{corollary}%\label{cor:eq}
Every label-stabilizing protocol that computes $\textsc{EQ}_n$ on the bidirectional $n$-ring has label complexity at least $\frac{n-2}{8}$.
\end{corollary}
\begin{proof}[Proof of Corollary \ref{cor:eq}]
Use $m=n/2$, $C=\{(n,1), (n/2+1, n/2)\}$, and $D=\{(1,n)(n/2,n/2+1)\}$. Define a fooling set of size $2^{n/2-1}$ by $S = \{ (x,x) : x\in \{0, 1\}^{n/2} \text{  and  }
x_1=1\}$.  
\end{proof}\\

\begin{corollary}%\label{cor:maj}
Every label-stabilizing protocol that computes $\textsc{Maj}_{n}$ on the bidirectional $n$-ring has label complexity at least $\log(\lfloor n/2\rfloor)/4$.
\end{corollary}
\begin{proof}[Proof of Corollary \ref{cor:maj}]
Use $m=\lfloor n/2\rfloor$, $C=\{(n,1), (n/2+1, n/2)\}$, and $D=\{(1,n)(n/2,n/2+1)\}$. 

Define: 
\[Q = \{ (1, 1^{k}0^{m-1-k})\}_{k=0}^{m-1}\,.\] 
The fooling set for odd sized rings is,  
\[S =\{ (x, \bar{x}, 1) | x\in Q \}\,.\] 
And for even sized rings,
\[S =\{ (x, \bar{x}) | x\in Q \}\,,\] 
where $\bar{x}_i = 1-x_i$ for every $i$. 
The size of the fooling set is $|S| = m = \lfloor n/2\rfloor$.
	 
\end{proof}

\end{document}

%% file: header.tex
\usepackage{graphicx,subfigure}
\usepackage{epsfig}
\usepackage[hidelinks]{hyperref}
\usepackage{amsmath}
\usepackage{amssymb}
\usepackage{algorithm}
\usepackage{algorithmic}
\usepackage{url}
\usepackage{enumerate}
\usepackage{amsfonts}
\usepackage{boxedminipage}
\usepackage{xcolor}
 \usepackage{framed}  
\usepackage{rotating}
\usepackage{array}
\usepackage{multirow}
\usepackage{color}
\usepackage{tikz}
\usepackage{mathabx}
\usepackage{cleveref}
\usepackage{tabularx,ragged2e,booktabs,caption}

\newenvironment{myalgo}[1]%
{
\begin{center}
\begin{boxedminipage}{0.8\linewidth}
\begin{center}
\textbf{\texttt{#1}}
\end{center}
\rm
\begin{tabbing}
....\=...\=...\=...\=...\=  \+ \kill
} %
{\end{tabbing} 
\end{boxedminipage} \end{center} %\vspace{0.3cm}
}

{
\vspace{0.01cm}
\begin{center}
\begin{boxedminipage}{0.8\linewidth}
\begin{center}
\textbf{\texttt{#1}}
\end{center}
} %
{
\end{boxedminipage} \end{center} \vspace{0.3cm}
}

\newtheorem{theorem}{Theorem}[section]
\newtheorem{definition}[theorem]{Definition}
\newtheorem{lemma}[theorem]{Lemma}
\newtheorem{corollary}[theorem]{Corollary}
\newtheorem{proposition}[theorem]{Proposition}

\newtheorem{example}{Example}

\newtheorem{claim}[theorem]{Claim}

\newcommand{\Expect}{{\rm I\kern-.3em E}}

\newcommand{\bN}{\mathbb{N}}

\newenvironment{proof}{\par\noindent{\bf Proof\ }}{\hfill\BlackBox\\[2mm]}
\newcommand{\BlackBox}{\rule{1.5ex}{1.5ex}}

\makeatletter
\def\moverlay{\mathpalette\mov@rlay}
\def\mov@rlay#1#2{\leavevmode\vtop{%
   \baselineskip\z@skip \lineskiplimit-\maxdimen
   \ialign{\hfil$\m@th#1##$\hfil\cr#2\crcr}}}
\newcommand{\charfusion}[3][\mathord]{
    #1{\ifx#1\mathop\vphantom{#2}\fi
        \mathpalette\mov@rlay{#2\cr#3}
      }
    \ifx#1\mathop\expandafter\displaylimits\fi}
\makeatother

\newcommand{\namedref}[2]{\hyperref[#2]{#1~\ref*{#2}}}

\renewcommand{\eqref}[1]{Equation~(\ref{#1})}

\newcommand{\thmref}[1]{Theorem~\ref{#1}}

\newcommand{\logBOS}{\widetilde{\textup{OS}}^b_{\log}}

\newcommand{\ppoly}{\textup{P/poly}}
\newcommand{\lpoly}{\textup{L/poly}}

\newcounter{todocounter}
\newcommand{\todonum}{\stepcounter{todocounter}{(\thetodocounter)}}

\def\shownotes{1}   % set 1 for version with author notes
\ifnum\shownotes=1
\newcommand{\authnote}[2]{{ $\ll$\textsf{\footnotesize \todonum\  #1 notes:  #2}$\gg$}}
\else
\newcommand{\authnote}[2]{}
\fi